\documentclass[12pt]{article}
\usepackage{apacite,latexsym,graphics,epsfig,amsmath,amssymb, 
tabularx,booktabs, color, mathtools,tikz}
\usepackage{natbib}
\usepackage{enumerate}
\usepackage{amsthm}
\usepackage{url}
\usepackage{fancyhdr,titling}
\usepackage[mathscr]{eucal}
\usetikzlibrary{trees}

\AtEndDocument{\refstepcounter{equation}\label{finalequat}}

\tikzstyle{bag} = [align=center]

\usepackage[T1]{fontenc}
\theoremstyle{definition}
\newtheorem{definition}{Definition}[section]
\theoremstyle{remark}
\newtheorem{example}{Example}[section]
\theoremstyle{plain}
\newtheorem{theorem}{Theorem}[section]
\newtheorem{lemma}[theorem]{Lemma}

\newtheorem{corollary}[theorem]{Corollary}

\newcommand{\tr}{^{\prime}}

\makeatletter
\newcommand{\oset}[2]{%
  {\mathop{#2}\limits^{\vbox to -.1\ex@{\kern-\tw@\ex@
   \hbox{\scriptsize #1}\vss}}}}
   
\def\keywords#1{{\vskip4pt
\noindent
\hbox to50.5pt{KEYWORDS:\quad\hss}\vtop{\advance \hsize by -59.5pt
\leftskip=28pt \rightskip=0pt
\noindent\ignorespaces#1\vskip8pt}}}

\newcommand*\xbar[1]{%
  \hbox{%
    \vbox{%
      \hrule height 0.5pt 
      \kern0.5ex
      \hbox{%
        \kern-0.25em
        \ensuremath{#1}%
        \kern-0.1em
      }%
    }%
  }%
}

\makeatletter
\let\runauthor\@author
\let\runtitle 

\date{}

\begin{document}

\title{A geometric power analysis for general log-linear models}

\author{Anna Klimova \\
{\small{National Center for Tumor Diseases (NCT), Partner Site Dresden, and}}\\
{\small{Institute for  Medical Informatics and Biometry,}}\\ 
{\small{Technical University, Dresden, Germany} }\\
{\small \texttt{anna.klimova@nct-dresden.de}}\\
{}
}

\maketitle

\begin{abstract}
Log-linear models are widely used to express the association in multivariate frequency data on contingency tables. The paper focuses on the power analysis for testing the goodness-of-fit hypothesis  for this model type. Conventionally, for the power-related sample size calculations a deviation from the null hypothesis (effect size) is specified by means of the chi-square goodness-of-fit index.  It is argued that the odds ratio is a more natural measure of effect size, with the advantage of having a data-relevant interpretation. Therefore,  a class of log-affine models that are specified by odds ratios whose values deviate from those of the null by a small amount can be chosen as  an alternative. Being expressed  as sets of constraints on odds ratios, both hypotheses are represented by smooth surfaces in the probability simplex, and thus, the power analysis can be given a geometric interpretation as well. A concept of geometric power  is introduced and a Monte-Carlo algorithm for its estimation is proposed. The framework is applied to the power analysis of goodness-of-fit in the context of multinomial sampling. An iterative scaling procedure for generating distributions from a log-affine model is described and its convergence is proved. To illustrate, the geometric power analysis is carried out for data  from a  clinical study.   \end{abstract}


\baselineskip=18pt

\section{Introduction}

Log-linear models are frequently employed in the social sciences, machine learning, iterative image reconstruction, and natural language processing. Among them are the conventional (hierarchical) log-linear models for contingency tables \citep*{BFH},  topological models \citep{Hauser1978}, non-standard log-linear models \citep{Rindskopf1990}, Markov random fields  \citep*{Malouf, Huang2010},  indicator models \citep{ZeltermanBook},  and relational models \citep*{KRD11, Forcina2019}. Several further generalizations, known as the generalized log-linear models (GLLM), were described in \cite{ThompsonGLLM1981},  \cite{EspelandGLLM1986},  \cite{kateri2014},  among others.  A log-linear model is a set of distributions on a discrete finite sample space, for example, a contingency table,  whose parameters, $\boldsymbol p$, can specified as 
\begin{equation}\label{mainModel}
\log \boldsymbol{p} = \mathbf{A} \boldsymbol{\beta},
\end{equation}
where $\mathbf{A}$ is a design matrix with linearly-independent columns,  and  $\boldsymbol \beta$ is a vector of so-called log-linear parameters.  The structure of $\mathbf{A}$ is entailed by the model-specific constraints.  For many model types,  such as  conventional log-linear models, $\mathbf{A}$ is a 0-1 matrix.  This paper focuses on the class of models, referred to as general log-linear by \cite{Haber}, which can be (re-)parameterized in a way that their design matrix has exclusively non-negative integer entries.  Examples can be found among the probabilistic trees and monomial models \citep{GeorgenSmithTrees2018, LeonelliRico}.   In this manuscript,  for simplicity of presentation,  the models in this generality will be addressed as log-linear,  although the level of generalization considered by \cite{Haber} will be assumed throughout.  

This paper discusses several aspects of the power analysis for goodness-of-fit testing under log-linear models.  In applied research, for example in clinical studies, two kinds of power analyses are routinely performed.  Apriori power analysis is conducted before a study begins and aims at finding a sample size which would allow to reach a nominal power, often $80\%$, in testing the study's primary hypothesis.  Such calculations are especially relevant for planning prospective clinical trials where the study size may need to be specified in advance.  After a study is complete,  a power analysis is also sometimes carried out to estimate the achieved (posteriori) power in testing statistical significance under the observed data. If the null hypothesis concerns the goodness of fit of a log-linear model, then, for the purpose of power calculations, the actually observed frequency distribution which is  used to test the goodness-of-fit  is assumed to be the alternative model.  The power is then calculated using a noncentral chi-square distribution whose noncentrality parameter is equal to the observed value of the Pearson chi-square statistic normalized by the observed total. This parameter characterizes the discrepancy between the null and alternative hypotheses and can be given an effect size interpretation.  

The apriori power analysis poses more of a challenge with respect to the choice of an alternative model.  One option is to use a hypothetical frequency distribution that seems to be a plausible alternative and calculate the resulting test power against the null model. Several examples illustrating this approach for hierarchical log-linear models were given  by \cite{OlerNoncentrality}.  The second option is to conduct power calculations by specifying the noncentrality parameter directly, relating it to an omnibus value of the goodness-of-fit index \citep{Cohen1988}.  Apart from the fact that it can be hard either to justify a specific choice of  an alternative distribution or to interpret a noncentrality parameter in data-relevant terms,  a more serious concern can be raised. In the analyses of categorical variables,  an effect-size interpretation is routinely given to the proximity between an odds ratio and the unity   \citep*[cf.][]{Chen2010big}.  However,  it can happen that the interpretation of a strength of association  in terms of odds ratios does not agree with the interpretation implied from the goodness-of-fit test.  As illustrated by \cite*{HaddockOR} for the model of independence on a $2\times 2$ table, there exist  distributions which are close to independence if the proximity is measured by the odds ratio but far from it if the deviation is assessed using the chi-square test.  One can expect that this inconsistency persists when larger tables and more complex models are involved.  This paper aims to address this issue and proposes a geometric framework for the power analysis of log-linear models. 

Assume that the null hypothesis of interest is posed as a log-linear model,  $\mathcal{G}_{0}$, and  tested against a log-affine model, $\mathcal{G}_{1}(\boldsymbol \xi)$, with the same design matrix:
$$\mathcal{G}_0: \, \log \boldsymbol \pi = \mathbf{A}\boldsymbol \beta \,\quad \mbox{ vs }\, \quad \mathcal{G}_1(\boldsymbol \xi): \, \log \boldsymbol \pi = \mathbf{A}\boldsymbol \beta + \log \boldsymbol \xi,$$
for a given offset-vector $\boldsymbol \xi$.   Recall that a set of restrictions specifying a log-linear or log-affine model defines a smooth surface in the probability simplex $\Delta_I$ of an appropriate dimension $I$ \citep{SturBook, SullivantAS}.  Then, because a goodness-of-fit chi-square statistic (for example the Pearson chi-square) may be interpreted as a squared distance between a distribution and a log-linear model,  a significance testing procedure can be given a geometric interpretation \citep[cf.][]{DiaconisEfron1985}.  The idea of geometric power analysis developed in this manuscript is the following. Take  $\epsilon > 0$ and surround the surface $\mathcal{G}_0$ by a ``tube''  consisting of the probability distributions in $\Delta_I$ whose goodness-of-fit statistic under $\mathcal{G}_0$ is less than $\epsilon$. Call this tube  an ``acceptance region'' of radius $\epsilon$, and its complement in $\Delta_I$ a ``rejection region''. The relative surface area of the subset of $\mathcal{G}_{1}(\boldsymbol \xi)$ outside the tube  with respect to the whole surface area of $\mathcal{G}_1(\boldsymbol \xi)$ is defined as geometric power.  The geometric power is thus equivalent to the probability of rejection of $\mathcal{G}_0$ in favor of $\mathcal{G}_1(\boldsymbol \xi)$ assuming the latter holds.


Notice that the concept of geometric power describes a relationship between two surfaces in a probability simplex but lacks the stochastic component. In particular, the definition does not take into account that the data  represent a realization of a discrete random variable.  
For the power calculations in this manuscript, it will be assumed that the data come from a multinomial distribution. In this case the sample size is fixed in advance,  so the apriori power calculations are more relevant. A Monte-Carlo-based stochastic extension of the geometric power  is defined for a given pair of $\mathcal{G}_0$ and  $\mathcal{G}_1(\boldsymbol \xi)$ and a fixed sample size $N$.  The multinomial sampling scheme is applied to draw a frequency distribution parameterized by $N$ and a probability distribution from  $\mathcal{G}_1(\boldsymbol \xi)$. The goodness-of-fit of the obtained realization with respect to the null model can be assessed using the chi-square test under a prespecified significance level, so the corresponding acceptance radius is equal to a  quantile of the relevant chi-square distribution. The procedure is replicated by sampling different probability distributions from  $\mathcal{G}_1(\boldsymbol \xi)$ and obtaining a multinomial realization from each. Finally, the proportion of replicates in the Monte-Carlo sample which were rejected by the chi-square test is computed. Because a multinomial  distribution can result in multiple stochastic realizations as frequency distributions, the empirically obtained  stochastic extension is referred to as a cumulative geometric power.   
The cumulative power analysis can be implemented for a range of sample sizes to obtain a power  table that can be used to select a sample size sufficient to achieve a desired power in the goodness-of-fit testing.

Whereas the geometric or cumulative power for a given pair of models can be approximated using a Monte-Carlo procedure in a straightforward way,  the generation of probability distributions from an alternative log-affine model is non-trivial. The construction algorithm proposed here utilizes the fact that a log-affine model is an exponential family and thus induces a mixed parameterization of distributions in the probability simplex $\Delta_I$.  Given a matrix $\mathbf{D}$ formed by a basis of the kernel space of $\mathbf{A}$, each  $\boldsymbol p \in \Delta_I$ is uniquely specified by its mean-value parameters $\mathbf{A}\tr\boldsymbol p$ and its canonical parameters $\mathbf{D}\log \boldsymbol p$.  Under such parameterization, the canonical parameters of distributions in  $\mathcal{G}_1(\boldsymbol \xi)$ stay constant and equal to $\log \boldsymbol \xi$. The model generation is carried out in two steps: (a) selection of values for the mean-value parameters $\mathbf{A}\tr\boldsymbol p$,  and (b) computation of the unique distribution parameterized by the selected values and the canonical parameters $\log \boldsymbol \xi$.  The procedure is essentially the MLE computation under $\mathcal{G}_1(\boldsymbol \xi)$. Step (a) can be implemented, for example, by using the mean-value parameters from a randomly drawn point from the probability simplex $\Delta_I$. Step (b) is a classical setup for using iterative scaling.  In fact, the original iterative proportional fitting  (IPF) algorithm for  classical hierarchical log-linear models  has a clear geometric interpretation in terms of projections with respect to the Kullback-Leibler information divergence  \citep{Csiszar}.  An analogous interpretation can be given to the generalization of the IPF called the generalized iterative scaling (GIS) of \cite{DarrochRatcliff}, see  \cite{CsiszarGIS}.  However,  one of the main assumptions standing behind the convergence of IPF and GIS, namely, the presence of the genuine overall effect, prevents one from using these algorithms for  log-linear models in which no such effect is intrinsically possible.  The iterative scaling algorithm proposed here can be applied in the no-overall effect case as well and is shown to generate a sequence of projections with respect to the Bregman information divergence \citep{Bregman} which converges to the unique distribution in $\mathcal{G}_1(\boldsymbol \xi)$  with prespecified values of $\mathbf{A}\tr\boldsymbol p$.   The proposed algorithm and its proof generalize the GIS itself and its geometric interpretation given by \cite{CsiszarGIS}. Finally, although a  procedure with a geometric meaning might be better suited for a geometry-based framework,  step (b) can, of course, be performed using a version of Newton's algorithm \citep[cf.][]{kateri2014,VermuntBook1997}, or  methods of convex optimization \citep[cf.][]{BertsekasNLP}. However, in these methods, the values for canonical parameters would be reestimated at each iteration, and, as a result, the limiting distribution would be close to the model of interest $\mathcal{G}_1(\boldsymbol \xi)$ only up to a given precision.

The manuscript is structured as follows. In Section \ref{sectionNotation}, the definition of log-linear models and their properties relevant for the power analysis are reviewed.  Section \ref{SectionPower} introduces a geometry-based framework for the power analysis of goodness of fit and proposes the concepts of geometric power and cumulative geometric power. The respective Monte-Carlo algorithms for the power estimation are described. In Section \ref{MLEcomputeMine}, an iterative scaling procedure  for generating distributions under the alternative hypothesis of Section \ref{SectionPower} is proposed and its convergence is proved. In Section \ref{SectionExample},  the framework is applied to carry out the geometric power analysis for a clinical study. Concluding remarks appear in Section \ref{SectionConclusion}.

\section{Background and notation}\label{sectionNotation}

Let $\boldsymbol Y$ be a discrete random vector  on a finite sample space  $\mathcal{I}$ which is treated as an ordered sequence of  $I = |\mathcal{I}|$ elements called cells.   Suppose the distribution of $\boldsymbol Y$ is parameterized by $\boldsymbol \delta =(\delta_i)_{i = 1}^I \in \mathbb{R}_{>0}^I$, where $\boldsymbol \delta$ is either a probability vector, $\boldsymbol \delta \equiv \boldsymbol p$, with $\sum_{i = 1}^I p_i = 1$, or a vector of intensities, $\boldsymbol \delta \equiv \boldsymbol \lambda$. Write $\mathcal{P}$ for the set of all strictly positive distributions on $\mathcal{I}$ and $\Delta_I$ for the open $(I-1)$ dimensional simplex.

Let $\mathbf{A} = (a_{ij}) \in \mathbb{Z}_{\geq 0}^{I \times J}$ be an $I \times J$ matrix with linearly-independent columns. The general log-linear model $\mathcal{G}(\mathbf{A})$ is the subset of $\mathcal{P}$ that satisfies:
\begin{equation} \label{PMmatr}
\mathcal{G}(\mathbf{A}) = \left \{ \boldsymbol \delta \in \mathcal{P}: \,\, \mbox{log } \boldsymbol \delta = \mathbf{A}\boldsymbol \beta, \, \mbox{for some } \, \boldsymbol \beta \in \mathbb{R}^J \right\}.
\end{equation}
Here, the components of $\boldsymbol \beta$ are the log-linear parameters of the model.  Equivalently, the model can be expressed either as an exponential family or in a multiplicative form:
\begin{equation} \label{RMexpF}
\mathcal{G}(\mathbf{A}) = \{ \boldsymbol \delta \in \mathcal{P}: \,\,\delta_i = \exp\{\sum_{j=1}^J  a_{ij} \beta_j\} = \prod_{j=1}^J \theta_j^{a_{ij}}, \, i \in \mathcal{I}, \, \mbox{ for some } \, \boldsymbol \theta \in \mathbb{R}^J_{>0}\},
\end{equation}
where $\theta_j = \mbox{exp  }(\beta_j)$, and $\boldsymbol \theta =(\theta_1, \dots, \theta_J) \in \mathbb{R}_{>0}^J$ denotes the vector of (multiplicative) parameters associated to columns of $\mathbf{A}$.  The column span of $\mathbf{A}$ is called the design space, and its co-dimension gives the model's degrees of freedom $K = dim(Ker(\mathbf{A})) = I - J$.  For example, in conventional log-linear models,  the model matrix $\mathbf{A}$ comprises the indicators of cylinder sets of marginal distributions of $\boldsymbol{Y}$, and the components of $\boldsymbol \beta$ are the interaction parameters associated to these sets \citep{Agresti2002}. In the framework of relational models,  the columns of $\mathbf{A}$ are indicators of arbitrary subsets of $\mathcal{I}$, with $\boldsymbol \beta$ being the subset parameters \citep{KRD11}.  In the case of probabilistic staged trees, the columns of $\mathbf{A}$ reflect on root-to-leaf paths,  and $\boldsymbol \theta = \exp(\boldsymbol \beta)$ correspond to tree edges \citep{GeorgenSmithTrees2018}.

Routinely, the unitary column $\boldsymbol 1 = (1, \dots, 1)\tr$ is assumed to belong to the design space. If this holds,  $\mathcal{G}(\mathbf{A})$ can be reparameterized to have a parameter, referred to as the overall effect, that is common to all cells in $\mathcal{I}$. In the models for probabilities, the presence of the overall effect implies that the normalization constraint $\boldsymbol 1\tr \boldsymbol p = 1$ can be tailored in the log-linear representation $ \mbox{log } \boldsymbol \delta = \mathbf{A}\boldsymbol \beta$.  However, if $\boldsymbol 1 \notin colspan (\mathbf{A})$, that is, no genuine overall effect is present, the normalization has to be imposed in addition to the log-linear constraints. 
 
A dual representation of the model (\ref{PMmatr}) can be obtained using a $K \times I$ matrix, $\mathbf{D}$, comprised by a basis of the kernel space $Ker(\mathbf{A})$. Because $\mathbf{D} \mathbf{A} = \mathbf{0}$, one has: 
\begin{equation}\label{dualF2d}
\mathcal{G}(\mathbf{A}) = \{ {\boldsymbol \delta} \in \mathcal{P}: \,\, \mathbf{D} \mbox{log }\boldsymbol \delta = \boldsymbol 0\}.
\end{equation}
For example, if $\mathcal{G}(\mathbf{A})$ is a hierarchical log-linear model, one can find a dual representation consisting entirely of marginal and conditional (log) odds ratios \citep{Agresti2002}.  

Being an exponential family, a log-linear model induces a mixed parameterization of all positive probability distributions  $\boldsymbol \delta \in \mathcal{P}$, where $\mathbf{A}\tr \boldsymbol \delta$ and $\mathbf{D} \log \boldsymbol \delta$ play the roles of mean value and canonical parameters, respectively. The variation independence and dual representation (\ref{dualF2d}) imply that for all $\boldsymbol \delta \in \mathcal{G}(\mathbf{A})$ the canonical parameters equal to zero and the distributions are uniquely specified by their mean-value parameters. As an illustration, take a probability distribution $\boldsymbol p = (p_{00}, p_{01}, p_{10}, p_{11})\tr$ on a $2\times 2$ contingency table. The row- and column  sums $p_{0+} = p_{00} + p_{01}$, $p_{1+} = p_{10} + p_{11}$, $p_{+0} = p_{00} + p_{10}$, $p_{+1} = p_{01} + p_{11}$ are the mean-value parameters entailed by the two-way independence model and are variation independent from the odds ratio  $OR = p_{00}p_{11} / (p_{01}p_{10})$. All distributions in the model of independence satisfy $OR = 1$ and can be reproduced uniquely by prescribing $p_{0+}$, $p_{1+}$, $p_{+0}$, $p_{+1}$ \citep[cf.][]{BFH}. 

The main points about the maximum likelihood estimation under log-linear models are outlined next. 
Assume that $\boldsymbol{Y}$  has either a multivariate Poisson distribution $Pois(\boldsymbol \lambda)$ or a multinomial distribution $Mult(N, \boldsymbol p)$, and let $\boldsymbol y$ be an observed realization of $\boldsymbol{Y}$.  
If the MLE of $\boldsymbol \lambda$ or $\boldsymbol p$ exists, it is the unique maximum of the Poisson or multinomial likelihood function, respectively. Implementing the Lagrange multiplier method, one can show that the MLE is the unique solution to the system of equations \citep{SilveyLMtest, AitchSilvey60}:
\begin{align}
&\mathbf{A}\tr\hat{\boldsymbol{\lambda}} = \mathbf{A}\tr \boldsymbol y, \quad \mathbf{D} \mbox{log } \hat{\boldsymbol{\lambda}} = \boldsymbol 0 \quad \mbox{ (Poisson case)}; \label{MLEsysP}\\
&\mathbf{A}\tr\hat{\boldsymbol{p}} = \gamma\mathbf{A}\tr (\boldsymbol y/N), \quad \mathbf{D} \mbox{log } \hat{\boldsymbol{p}} = \boldsymbol 0,  \,\,\boldsymbol 1'\boldsymbol p = 1 \quad \mbox{ (Multinomial case)}.  \label{MLEsysM}
\end{align}
The coefficient $\gamma$ appearing in (\ref{MLEsysM}) stands for a Lagrangian multiplier. If $\mathcal{G}(\mathbf{A})$ is a model for intensities, then $\gamma \equiv 1$, that is, the mean-value parameters (sufficient statistics) of the MLE are equal to those of the observed distribution \citep[cf.][]{KRD11}.  This property characterizing the observed and estimated mean-value parameters is sometimes referred to as the mean-value theorem or Birch theorem. If $\mathcal{G}(\mathbf{A})$ is a model for probabilities with the overall effect, the mean-value theorem holds as well \citep[cf.][]{Andersen74}, so $\gamma \equiv 1$. Moreover, the equality of sufficient statistics $\mathbf{A}\tr\hat{\boldsymbol{p}} = \mathbf{A}\tr (\boldsymbol y/N)$ automatically entails the normalization $\boldsymbol 1'\boldsymbol p = 1$.  However,  when $\mathcal{G}(\mathbf{A})$ is a model for probabilities without the overall effect, the Lagrangian multiplier depends on the data, $\gamma = \gamma(\boldsymbol y)$, and the sufficient statistics of the MLE are proportional, but not equal, to those observed. This result is a generalization of Birch theorem, and $\gamma$ is called the adjustment factor  \citep{KRD11, Forcina2019}.

A model without a genuine overall effect is used for a running example presented next. As emphasized, the absence of the overall effect is not a matter of inconvenience but rather an intrinsic feature of the model. 

\begin{figure}
\begin{center}
\includegraphics[scale=0.6]{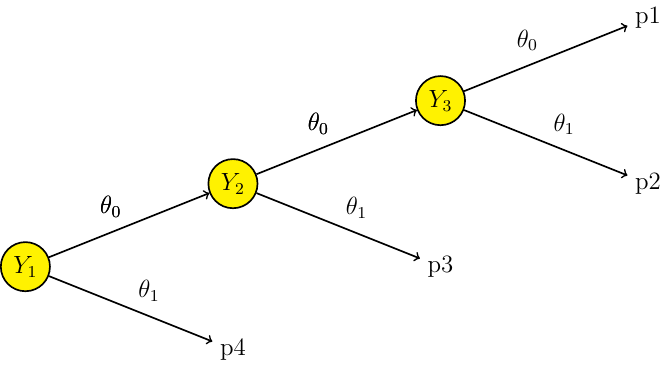}
\end{center}
\caption{Tree representation for the model in Example \ref{VaccineModel}}
\label{TreeVaccine}
\end{figure}

\begin{example}\label{VaccineModel}
Consider an experiment with repeated observations in which all participants are subsequently administered the same treatment with two outcomes, \textit{success} or  \textit{failure}. Suppose that, by design, the treatment can be reapplied to each participant up to three times but is stopped as soon as the first success is achieved. A clinical study where this design is relevant will be described in Section \ref{SectionExample}. A graphical representation  is given in Figure  \ref{TreeVaccine}, where $\theta_1$ denotes the probability of success, $\theta_0= 1-\theta_1$ the probability of failure, and $\boldsymbol p = $ ($p_1$, $p_2$, $p_3$, $p_4$) the probabilities of corresponding treatment paths. The model of independence of subsequent failures given the previous ones can be expressed either as a collection of risk equations that can be read out from  Figure \ref{TreeVaccine}:
\begin{equation}\label{calves3risk}
\frac{p_1 + p_2}{p_1 + p_2 + p_3} = \frac{p_1}{p_1 + p_2} = 1 - p_4, 
\end{equation} 
or in a multiplicative form:
\begin{equation}\label{calves3}
p_{1} = \theta_0^3, \,\,  p_{2} = \theta_0^2 \theta_1, \,\, p_{3} = \theta_0\theta_1, \,\, p_{4} = \theta_1.
\end{equation}
By substituting (\ref{calves3}) into the risk equations (\ref{calves3risk}), one can show that the condition $\theta_0 + \theta_1 = 1$ is necessary and sufficient for $\boldsymbol 1'\boldsymbol p = 1$.   A log-linear form of (\ref{calves3}) can be obtained using the design matrix 
\begin{equation}\label{VaccineMMx}
 \mathbf{A}\tr = \left( \begin{array}{cccc}
                                   3 & 2 & 1 & 0\\
                                   0 & 1 & 1 & 1\\
                                   \end{array} \right).
\end{equation}
Because $(1,1,1,1)\tr \notin colspan (\mathbf{A})$, the model (\ref{VaccineMMx})  does not have the genuine overall effect.  

Let $\boldsymbol y = (y_1, y_2, y_3, y_4)\tr$ represent the number of participants observed in each response category after the experiment is complete.   In the absence of drop-outs,  one can assume that $\boldsymbol y$  is a realization of $Mult(N, \boldsymbol p)$, where $N = \sum_{i = 1}^4 y_i$. The kernel of the multinomial log-likelihood is equal to
$L(\boldsymbol \theta \mid \boldsymbol y) = (3y_{1}+2y_{2} +  y_{3})\mbox{log } \theta_0  +  (y_{2}+y_{3} +  y_{4}) \mbox{log }\theta_1$, and  $S_1 = 3y_{1}+2y_{2} +  y_{3}$ and $S_2 = y_2+y_3 + y_4$ are the sufficient statistics of the model. The maximum likelihood estimators have a closed form, namely:
$\hat{\theta}_0 = S_1/ (S_1 + S_2)$, $\hat{\theta}_1 = S_2/ (S_1 + S_2)$, and
\begin{equation}\label{MLEcalves3}
\hat{\boldsymbol p} = \left (\frac{S_1^3}{(S_1 + S_2)^3}, \frac{S_1^2S_2}{(S_1 + S_2)^3}, \frac{S_1 S_2}{(S_1 + S_2)^2},  \frac{S_2}{(S_1 + S_2)}\right)\tr.
\end{equation}
Revisiting the system (\ref{MLEsysM}), notice that the adjustment factor $\gamma =  N  (\hat{\theta}_0^2 + \hat{\theta}_0 + 1)/{(S_1 + S_2)}$. \qed 
\end{example}
\noindent 

The goodness of fit testing for general log-linear models, whether or not the overall effect is present, can be carried out using the methods described in \cite{AitchSilvey58, AitchSilvey60} and \cite{SilveyLMtest}. When the model holds, the Pearson chi-square statistic $X^2$ and the likelihood ratio test statistic $G^2$: 
$$X^2(\boldsymbol y)=\sum_{i = 1}^I \frac{(y_i- \hat{y}_i)^2}{\hat{y}_i} \,\,\, \mbox{  and  }
G^2(\boldsymbol y)=2 (\sum_{i = 1}^I y_i \mbox{ log } ({y_i}/{\hat{y}_i}) - (y_i - \hat{y}_i))$$  
are asymptotically distributed as $\chi^2$ with $K = dim(Ker(\mathbf{A}))$ degrees of freedom.  Otherwise, when the model does not hold, $X^2(\boldsymbol y)$ and $G^2(\boldsymbol y)$ are distributed as a noncentral $\chi^2_K$ with parameter $\phi$ equal to 
\begin{equation}\label{NCP}
\phi = \sum_{i = 1}^I \frac{(y_i/N - \hat{y}_i/N)^2}{\hat{y}_i/N} =  1/N \cdot X^2(\boldsymbol y).
\end{equation}

The noncentrality parameter $\phi$ can be used to characterize a deviation of a hypothetical true distribution from the null model whose goodness-of-fit is tested, and the goodness-of-fit index computed as $w = \sqrt{\phi}$ is interpreted as an effect size \citep[cf.][]{Cohen1988}.  

A major caveat with interpretation of such an effect size will be illustrated next using the model of independence  on a $2 \times 2$ table.  The model is frequently used in practice, and a detailed computation of the MLE can be spared. 

\begin{table}[h]
\begin{minipage}[b]{0.45\linewidth}
\centering
\vspace{5mm}
\begin{array}[t]{l|cc}
& \multicolumn{2}{c}{Outcome}\\
\cline{2-3}
Treatment & Failure & Success\\
\hline \\ [-6pt]
Control & 1 & 9\\
Experimental & 9 & 33 \\[6pt]
\end{array}
\caption{A frequency table with $OR \approx 0.41$ \, ($1/OR \approx 2.45$) and $X^2 \approx 0.68$}
\label{ExTR1}
\end{minipage}
\hspace{0.5cm}
\begin{minipage}[b]{0.45\linewidth}
\centering
\vspace{5mm}
\begin{array}[t]{l|cc}
& \multicolumn{2}{c}{Outcome}\\
\cline{2-3}
Treatment & Failure & Success \\
\hline \\ [-6pt]
Control & 3 & 7\\
Experimental  & 7 & 35 \\[6pt]
\end{array}
\caption{A frequency table with $OR \approx 2.11$ and $X^2 \approx 0.92$}
\label{ExTR2}
\end{minipage}
\end{table}

\begin{example}\label{PowExampleInd}
Denote by $\boldsymbol y_1 = (1,9,9,33)\tr$ and $\boldsymbol y_2 = (3,7,7,35)\tr$  the hypothetical observed data presented in Tables \ref{ExTR1} and \ref{ExTR2}, respectively. Assume that $\boldsymbol y_1, \boldsymbol y_2$ are realizations of $Mult(N, \boldsymbol \pi)$, where $N = 52$, and $\boldsymbol \pi$ stands for the true probability vector. Under the model of independence, the corresponding Pearson statistics equal $X^2(\boldsymbol y_1) \approx 0.68$ and $X^2(\boldsymbol y_2) \approx 0.92 $, and the noncentrality parameters computed from (\ref{NCP}) are: $\phi_1 = 0.68/52$, $\phi_2 = 0.92/52$. Assuming a $5\%$ significance level, the achieved (posteriori) power in rejecting the hypothesis of independence is equal, in the first case, to $82\%$ and, in the second, to $99\%$. The results of power analysis agree with the values of the goodness-of-fit index: the alternative distribution in the first case is closer to the model of independence than the distribution in the second case, and is thus more difficult to reject. However, a closer look at the data shows that if the deviation from independence is measured using the odds ratio then the conclusion should be exactly the opposite, because $OR(\boldsymbol y_1) \approx 0.41$ ($1/OR(\boldsymbol y_1) \approx 2.45)$ and $OR(\boldsymbol y_2) \approx 2.11$. The inconsistency is elaborated further. \qed
\end{example}

The phenomenon demonstrated in Example \ref{PowExampleInd} is, in fact, a consequence of the variation independence between the odds ratio and the marginal sums in a $2 \times 2$ table. The MLE under independence, and in turn the Pearson statistic $X^2$ and deviance $G^2$, are all functions of row-and column sums of the table. The variation independence entails that, depending on a choice of marginal distributions, the same odds ratio value can lead to different  $\chi^2$ goodness-of-fit indices, and vice versa, the same value of the goodness-of-fit index may correspond to a large variety of association structures expressed using odds ratios.  To obtain a more detailed illustration, we generated $10^4$ probability vectors $\boldsymbol p \in \Delta_4$ using the flat Dirichlet distribution. In Figure \ref{GOF1}, their goodness-of-fit indices $w = w(\boldsymbol p)$  with respect to the $2 \times 2$ independence are plotted against the odds ratios $p_{00}p_{11}/ (p_{01} p_{10})$.  For example, the odds ratios for the distributions $\boldsymbol p_1 =$ $(0.250365,$ $0.230925,$ $0.181938,$ $0.336772)\tr,$
$\boldsymbol p_2 =$ $(0.703505,$ $0.252487,$ $0.025576,$ $0.0184322)\tr$ are about $2$, while $w(\boldsymbol p_1) \approx 0.171$ and $w(\boldsymbol p_2) \approx 0.071$. On the other hand,  the distributions $\boldsymbol p_3 =$  $(0.589401,$ $0.13757,$ $0.077083,$ $0.195946)\tr$, and $\boldsymbol p_4 =$  $(0.425385,$ $0.012916,$ $0.288966,$ $0.272734)\tr$ both have $w \approx  0.5$, while $OR( \boldsymbol p_3) \approx 10.9$, and $OR( \boldsymbol p_4) \approx 31.1$. Several examples with real frequency data illustrating that the relationship between the Pearson chi-square statistic and the odds ratio is not monotone can be found in \cite{HaddockOR}.  

Further, the issue of multiple possible association structures leading to the same value of a goodness-of-fit statistic stays in the shadow during the apriori power analysis used for the power-related sample size determination.  
For a sample size calculation based on a noncentral $\chi^2$ distribution, one has to specify a noncentrality parameter $\phi$ which is supposed to reflect on the magnitude of deviation of a hypothetical true distribution from the model of interest, for example independence.  In practice, one can choose $\phi$  relating it to an omnibus value of a goodness-of-fit index \citep{Cohen1988}, although it would be difficult to argue that an ad hoc chosen noncentrality parameter is a practically relevant effect size. Another way to proceed is to calculate $\phi$ from a prespecified probability distribution, as was described  in \cite{OlerNoncentrality} for the conventional log-linear models on contingency tables. Such a distribution is constructed based on a set of restrictions on the odds ratios and some plausible values for marginals of the table. This procedure is more intuitive than the index-based one, especially because the proximity between an odds ratio and the unity is routinely given an effect-size interpretation, see \cite{Chen2010big}, among others. On the other hand, it does not take into account that the resulting value of $\phi$ will depend on the preselected marginals, so the issue brought up in Example \ref{PowExampleInd} persists. 

The power analysis framework proposed in the next section  allows for defining a discrepancy from a log-linear model using exclusively odd ratios while placing no restrictions on the marginal distributions. Being  ``marginal-free``, the proposed method takes into account the whole family of values of the goodness-of-fit test statistics which are entailed by the distributions with a given odds ratios structure.

\begin{figure}
\begin{center}
\includegraphics[scale=0.4]{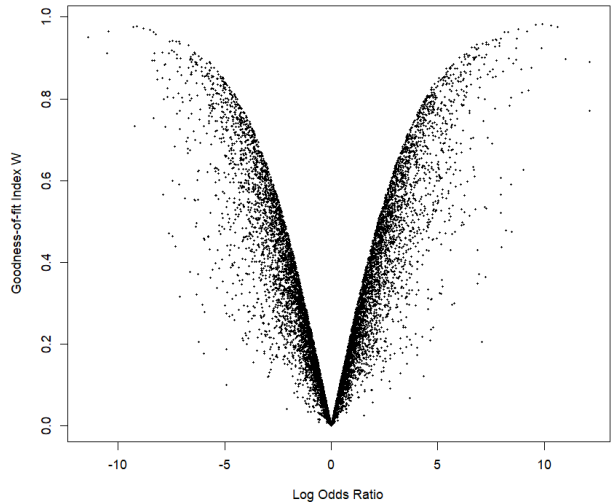}
\end{center}
\caption{Goodness-of-fit $w$-indices with respect to the model of independence on $2 \times 2$ contingency table, computed for $10^4$ probability distributions generated from the open simplex $\Delta_4$ using the Dirichlet distribution $Dir_4(1)$.}
\label{GOF1}
\end{figure}


\section{A geometric power of the chi-square test}\label{SectionPower}

 Let $\boldsymbol \pi \in \mathcal{P}$ denote the true distribution on $\mathcal{I}$ and $\mathcal{G}_{0}(\mathbf{A}) = \{\boldsymbol p \in \mathcal{P}: \, \mbox{log } \boldsymbol p = \mathbf{A}\boldsymbol \beta\}$ be a log-linear model which is believed to describe the association structure in $\boldsymbol \pi$. Suppose that the null hypothesis  
$\mathcal{H}_0: \boldsymbol \pi \in \mathcal{G}_{0}(\mathbf{A})$ is tested against a class of alternatives defined as a log-affine model:
$$\mathcal{H}_1: \,\,  \boldsymbol  \pi \in \mathcal{G}_{1}(\mathbf{A}, \boldsymbol \xi) = \{\boldsymbol p \in \mathcal{P}: \, \mbox{log } \boldsymbol p = \mathbf{A} \boldsymbol \beta + \log\boldsymbol \xi\},$$
for a fixed $\boldsymbol \xi \in \mathbb{R}^I_{>0}$.  Letting $\mathbf{D}$ be a kernel basis matrix of $\mathbf{A}$,  one can reformulate both  hypotheses in the dual form (\ref{dualF2d}):
\begin{align}
&\mathcal{H}_0: \boldsymbol  \pi \in \mathcal{G}_{0}(\mathbf{D}) = \{\boldsymbol p \in \mathcal{P}: \, \mathbf{D}\mbox{log } \boldsymbol p = \boldsymbol 0\},  \label{H0dual}\\
&\mathcal{H}_1: \boldsymbol  \pi \in \mathcal{G}_{1}(\mathbf{D}, \boldsymbol \xi) =  \{\boldsymbol p \in \mathcal{P}: \, \mathbf{D}\mbox{log } \boldsymbol p = \mathbf{D}\mbox{log } \boldsymbol \xi\},  \label{H1dual}
 \end{align}
In the sequel,  to unify and simplify the notation,  the models are referred to as $\mathcal{G}_{0}$ and $\mathcal{G}_{1}(\boldsymbol \xi)$. 

Notice that the alternative hypothesis (\ref{H1dual}) is expressed using only the canonical parameters $\mathbf{D}\mbox{log } \boldsymbol p$. Because of their variation independence from the mean-value parameters, the alternative does not incur any restrictions on $\mathbf{A}\tr \boldsymbol p$. Thus the discrepancy between the null and alternative hypotheses is quantified in terms of the canonical parameters only.

The statistical power of a hypothesis test characterizes the ability of the test to distinguish between the null and alternative hypotheses. For the goodness-of-fit test of $\mathcal{H}_0$ versus $\mathcal{H}_1$, the power is  the probability of rejection of the model $\mathcal{G}_0$ given that  the model $\mathcal{G}_1(\boldsymbol \xi)$ holds.  Because this definition refers to a hypothesis about two association structures rather than a hypothesis of goodness-of-fit of a particular distribution with respect to a certain model, a formal statistical power estimation would be very challenging.  A geometric framework for  power analysis of goodness-of-fit is introduced next. Looking at test power from a geometric perspective  highlights the intuition behind this concept and suggests a straightforward approach for the power estimation.

Take an $\epsilon > 0$ and a suitable goodness-of-fit statistic, for example the Pearson $X^2$. Denote by $\mathcal{T}_{0,\epsilon}$ the set of probability distributions in the simplex whose Pearson statistic $X^2$ computed under the model $\mathcal{G}_0$ is less than $\epsilon$:
\begin{equation}\label{tube0}
\mathcal{T}_{0,\epsilon} = \{\boldsymbol p \in \Delta_I:  \,\, X^2(\boldsymbol p) < \epsilon\}.
\end{equation}
In geometric terms $\mathcal{T}_{0,\epsilon}$ may be seen as a tube of radius $\epsilon$ around the surface $\mathcal{G}_{0}$. In statistical terms, one can say that this tube is comprised by distributions for which the null hypothesis $\mathcal{H}_0$ is not rejected. Thus $\mathcal{T}_{0,\epsilon}$ is an ``acceptance region'' of radius $\epsilon$, and the complement $\bar{\mathcal{T}}_{0,\epsilon} = \Delta_I \setminus {\mathcal{T}}_{0,\epsilon}$  is a ``rejection region''. 
Further, let 
\begin{equation}\label{tube1}
\mathcal{G}_{1,\epsilon}(\boldsymbol \xi) =   \mathcal{G}_{1}(\boldsymbol \xi) \cap \bar{\mathcal{T}}_{0,\epsilon} = \{\boldsymbol p \in \mathcal{G}_{1}(\boldsymbol \xi):  \,\, X^2(\boldsymbol p) \geq \epsilon \}.
\end{equation}
That is, $\mathcal{G}_{1,\alpha}(\boldsymbol \xi)$ consists of distributions in the model $\mathcal{G}_{1}(\boldsymbol \xi)$ for which the null hypothesis  $\mathcal{H}_0$ will be (correctly) rejected.  Both $\mathcal{G}_{1}(\boldsymbol \xi)$ and $\mathcal{G}_{1,\alpha}(\boldsymbol \xi)$ are surfaces in $\Delta_I$, with the latter being the proper subset of the former, and the ratio between their surface areas can be interpreted as the probability of rejection of $\mathcal{H}_0$ given $\mathcal{H}_1$ holds. Thus the following definition is proposed. 

\begin{definition}
Let $\mathcal{H}_0$  and $\mathcal{H}_1$  be the null and alternative hypotheses specified as a log-linear and log-affine models in (\ref{H0dual}) and (\ref{H1dual}), respectively. For a given $\epsilon > 0$ let $\mathcal{G}_{1,\epsilon}(\boldsymbol \xi)$ be the surface defined in (\ref{tube1}).  Then the ratio between the surface areas of $\mathcal{G}_{1,\epsilon}(\boldsymbol \xi)$ and $\mathcal{G}_{1}(\boldsymbol \xi)$,
$$Power_g  \coloneqq \frac{\mbox{Surface area of } \mathcal{G}_{1,\epsilon}(\boldsymbol \xi)}{\mbox{Surface area of } \mathcal{G}_{1}(\boldsymbol \xi)},$$
is called the geometric power of the goodness-of-fit test of $\mathcal{G}_{0}$ versus $\mathcal{G}_{1}(\boldsymbol \xi)$.
\end{definition}

\vspace{1mm}

In practice, $Power_g$ can be estimated using the following Monte-Carlo algorithm.  

\begin{center} \textbf{Monte-Carlo Geometric Power of $\mathcal{G}_{0}$ versus $\mathcal{G}_{1}(\boldsymbol \xi)$:} \end{center}
\begin{itemize}
\item[\textit{Input:}] an acceptance radius $\epsilon > 0$,  a number of replications  $N_{sim}$. 
\item For each $u \in 1, \dots, N_{sim}$, repeat \textit{Steps 1, 2}:
\item[\textit{Step 1 -}]  Generate a probability distribution $\boldsymbol \pi_u \in \mathcal{G}_1(\boldsymbol \xi)$. 
\item[\textit{Step 2 -}] Compute the MLE  $\hat{\boldsymbol \pi}_u$ under the null model $\mathcal{G}_0$ and collect the Pearson statistic $X^2(\boldsymbol \pi_u)$.
\item A Monte-Carlo estimate of geometric power is equal to:
$$\mbox{MC-}Power_g \coloneqq \frac{1}{N_{sim}} \sum_{u = 1}^{N_{sim}} \mathbb{I}(X^2(\boldsymbol \pi_u) \geq \epsilon ) = \frac{1}{N_{sim}} \sum_{u = 1}^{N_{sim}} \mathbb{I}(\boldsymbol \pi_u \in \mathcal{G}_{1,\epsilon}(\boldsymbol \xi)),$$
where $\mathbb{I}(\cdot)$ stands for the indicator function. \qed
\end{itemize}

\noindent \textit{Remark}: The geometric power can be computed with respect to the deviance statistic $G^2$ as well. This case is not considered for the sake of space.

\begin{example}\label{ExIndepGeom}
To illustrate the geometric power calculations, the model of independence for a $2 \times 2$ contingency table is revisited. In the dual representation,  this model can be formulated using the odds ratio constraint $p_{00}p_{11} / (p_{01}p_{10}) = 1$. As an alternative $\mathcal{G}_1(\boldsymbol \xi)$ consider the set of distributions  whose odds ratio is equal to $p_{00}p_{11} / (p_{01}p_{10}) = \xi$ for a given $0 < \xi \neq 1$.  Because of the variation independence between the odds ratio and the marginal distributions, the alternative model specification does not result in any restrictions on the row- or column sums $p_{0+}$, $p_{1+}$, $p_{+0}$, $p_{+1}$ \citep[Chapter 6]{TRbook2018}.  The geometric power calculations were performed for three alternatives, $\boldsymbol \xi = 1/10; \, 5; \, 50$, and for the acceptance radii $\epsilon \in (0.1, 0.4)$. The distribution generation in \text{Step 1} were performed in two steps: (1a) sampling a point $\boldsymbol p = (p_{00}, p_{01}, p_{10}, p_{11})\tr \in \Delta_4$; (1b) applying the iterative proportional fitting (IPF) procedure to find the distribution $\boldsymbol \pi = (\pi_{00}, \pi_{01}, \pi_{10}, \pi_{11})\tr \in \Delta_4$ such that $\pi_{0+} = p_{0+}$, $\pi_{0+} = p_{1+}$, $\pi_{0+} = p_{+0}$, $\pi_{0+}= p_{+1}$ and  $\pi_{00}\pi_{11} / (\pi_{01}\pi_{10}) = \xi$. In fact, the resulting distribution $\boldsymbol \pi$ is equal to the MLE under the model $\mathcal{G}_1(\boldsymbol \xi)$. The MLE in \text{Step 2} can be found in closed form \citep[cf.][]{BFH}. The resulting power curves are shown in Figures \ref{PowerIndepUniform} and \ref{PowerIndepJeffreys}. One can see that an increase of acceptance radius would make the rejection region smaller, and thus would lead to a smaller rejection probability (geometric power). Given the same acceptance radius, the rejection rate would be higher for an alternative which is further from the null model ($\boldsymbol \xi = 1$). That is, the geometric power increases as the effect size gets bigger. \qed
\end{example}

The proposed definition of geometric power describes a relationship between two surfaces in the probability simplex but does not take into account that the data on $\mathcal{I}$  represent a realization of a discrete random variable and, in fact, consists of absolute frequencies. In the sequel, unless specified otherwise, it is assumed that the data are collected using the multinomial sampling scheme, so the sample size $N$ is fixed in advance.  The Monte-Carlo based stochastic extension of the power concept introduced below describes a relationship between the model of interest and a realization of a multinomial distribution parameterized by $N$ and  some  $\boldsymbol p \in \mathcal{G}_1(\boldsymbol \xi)$.  Because for multinomial data the Pearson statistic of the MLE is asymptotically chi-squared as $N\to \infty$, the acceptance radius can be naturally specified using a quantile of the appropriate central chi-square distribution,  $\chi^2_{0; 1-\alpha}$, for some $\alpha \in (0,1)$. Essentially, a geometric goodness-of-fit can be assessed using the conventional chi-square test with $\alpha$ being the significance level.  In order to emphasize that multiple drawings from the same multinomial distribution can result in different stochastic realizations as frequency distributions, the MC-obtained  stochastic extension defined here  is referred to as the cumulative geometric power.   

\vspace{1mm}

\begin{center} \textbf{Monte-Carlo Cumulative Geometric Power of $\mathcal{G}_{0}$ versus $\mathcal{G}_{1}$:} \end{center}
\begin{itemize}
\item[\textit{Input:}] a sample size $N$,  a critical value $\alpha \in (0,1)$,  $\chi^2_{0; 1-\alpha}$,  a number of replications  $N_{sim}$. 
\item For each  $u \in 1, \dots, N_{sim}$ repeat \textit{Steps 1, 2, 3}:
\item[\textit{Step 1 -}]  Generate a probability distribution $\boldsymbol \pi_u \in \mathcal{G}_1(\boldsymbol \xi)$.
\item[\textit{Step 2 -}] Obtain a frequency distribution $\boldsymbol f_u$ as a realization of $Mult (N,\boldsymbol \pi_u)$.
\item[\textit{Step 3 -}] Compute the MLE  $\hat{\boldsymbol f}_u$ under the null model $\mathcal{G}_0$ and collect the Pearson statistic $X^2(\boldsymbol f_u)$.
\item A Monte-Carlo estimate of the power of the goodness-of-fit hypothesis test is computed as the empirical rejection rate, namely:
\begin{equation}\label{cumPow}
\mbox{MC-}Power_c \coloneqq \frac{1}{N_{sim}} \sum_{u = 1}^{N_{sim}} \mathbb{I}(X^2(\boldsymbol f_u) \geq \chi^2_{0; 1-\alpha}).
\end{equation}
\qed
\end{itemize}

The cumulative power analysis is relevant, for example, for creating the sample size-power tables,  used for selecting a sample size needed to reach a prespecified power  in testing (\ref{H0dual}) versus (\ref{H1dual}).  
In addition,  a cumulative analysis can be performed posteriori, that is,  to estimate a cumulative  power achieved by testing an observed frequency distribution, say $\boldsymbol f_0$.  For posteriori analysis,   the exponents of the canonical parameters of  $\boldsymbol p_0 = \boldsymbol f_0/(\boldsymbol 1\tr \boldsymbol f_0)$,  are used as the offset, $\boldsymbol \xi =   \exp \{\mathbf{D}\log \boldsymbol p_0\}$, and $N = \boldsymbol 1\tr \boldsymbol f_0$ is the observed total.

Both Monte-Carlo algorithms described in this section are relatively straightforward to implement.  The multinomial sampling in \textit{Step 2} can be performed using a statistical software. Numerous options are available for the MLE computation in \textit{Step 3}.  However, the distribution generation in \textit{Step 1} is quite non-trivial for a general log-linear model. Following the exponential family  approach used in Example \ref{ExIndepGeom}, the construction procedure can be conducted in two steps:
(1a) draw a $\boldsymbol p \in \Delta_{I}$ and compute the mean-value parameters $\mathbf{A}\tr\boldsymbol p$,  (1b) determine the unique distribution whose mean-value parameters are $\mathbf{A}\tr\boldsymbol p$ and the canonical parameters $\log \boldsymbol \xi$.  Equivalently, in order to generate a distribution $\boldsymbol \pi \in \mathcal{G}_1(\boldsymbol \xi)$, draw a $\boldsymbol p \in \Delta_{I}$,  compute the MLE $\hat{\boldsymbol p}$ of $\boldsymbol p$ under the model $\mathcal{G}_1(\boldsymbol \xi)$ and set  $\boldsymbol \pi = \hat{\boldsymbol p}$. Step (1a) can be implemented by using the mean-value parameters from a randomly drawn point from the probability simplex $\Delta_I$.  For instance, in the simulations in Example \ref{ExIndepGeom},  the Dirichlet distribution ${Dir}_I(\boldsymbol \alpha)$,  where $\boldsymbol \alpha = (1, \dots, 1)$ (uniform distribution on $\Delta_I$) or $\boldsymbol \alpha = (1/2, \dots, 1/2)$ (Jeffreys prior on $\Delta_I$) was used. The setup in Step (1b) naturally calls for using iterative scaling.   Moreover, the original iterative proportional fitting algorithm for  classical hierarchical log-linear models and its  generalization called the generalized iterative scaling  \citep {DarrochRatcliff} both have a clear geometric interpretation in terms of projections with respect to the Kullback-Leibler information divergence  \citep{Csiszar,  CsiszarGIS}.  The iterative scaling approach to constructing distributions from a log-linear model is considered in detail in the following section.

\section{Iterative scaling for sampling from a log-linear model}\label{MLEcomputeMine}

The classical Iterative Proportional Fitting (IPF) is a well-known procedure for the maximum likelihood estimation under traditional log-linear models, see, for example, \cite{Agresti2002} and references therein. Numerous extensions of IPF were developed in order to incorporate the MLE computation under different kinds of non-standard log-linear models, in particular, those allowing for non-negative but not necessarily 0-1 design matrices. Firstly, it was the GIS of \cite{DarrochRatcliff} and, later, its further generalizations, see   \cite{Herman1976iterative, Censor1987block, ByrneAlgorithms, Holte1990iterative,  Lent1991primal, WinklerDykstras}, among others. The Multiplicative Algebraic Reconstruction Techniques (MART), widely used in image reconstruction, is also a special type of iterative scaling  \citep[cf.][]{ByrneAlgorithms}.  Among other usages of IPF is a  population synthesis which aims at generating frequency distributions with a given set of marginals and a particular interaction structure. For example, in the small area estimation studies the IPF may be used to calculate distributions whose interactions coincide with those in the population while the marginals are retained from a specific sub-population.  Therefore, using the iterative scaling for the distribution generation in \textit{Step 1b}  would be a natural choice. However, a closer look at the algorithms mentioned above brings attention to the following three points. 

The first point concerns the fact that the proof of convergence of the original GIS was given under the assumption that all sums along columns of the design matrix  equal $1$:
\begin{equation}\label{sumTo1}
\sum_{j = 1}^J a_{ij}= 1, \quad \mbox{ for each } \,\, i \in 1, \dots, I.
\end{equation} 
In the studies of convergence of iterative scaling, the condition (\ref{sumTo1}) or its more general version $\boldsymbol 1\tr \mathbf{A} = C\boldsymbol 1\tr$, where $C = const > 0$, are almost always imposed. Some authors do it directly, among those are \cite{Herman1976iterative, Censor1987block, ByrneAlgorithms, Holte1990iterative, Byrne2009block, Byrne1997convergent, Abou2022image, WinklerDykstras}, while others, for example, \cite{Lent1991primal}, do not mention explicitly but require in the proofs of convergence.  Our extensive literature search on MART-type algorithms did not find any examples where the assumption $\boldsymbol 1 \in colspan (\mathbf{A})$ was completely relaxed, that is, not being used for the convergence proofs. Even a very thorough account of the simultaneous MART (SMART) given in \cite{Byrne2009block, ByrneAlgorithms} did not discuss any design matrices where the overall effect was not present. Moreover, the presence of the overall effect by design is simply anticipated in the image reconstruction context, where MART is employed. The second point is, because the assumption (\ref{sumTo1})  is necessary and sufficient for the presence of the genuine overall effect, $\boldsymbol 1 \in colspan (\mathbf{A})$, neither GIS nor (S)MART is suitable for the general log-linear models without the overall effect, for instance, the one in Example \ref{VaccineModel}.  
As the third point, notice that \cite{KRipf1} already developed an algorithm specifically designed to handle the absence of the overall effect. However, their proof of convergence was given for 0-1 design matrices and cannot be extended to those with non-negative integer entries \citep{KRipf1}, such as  (\ref{VaccineMMx}) in Example \ref{VaccineModel}.

The algorithm IPF($\gamma$, $\boldsymbol \xi$) proposed below is, on one hand, a generalization of Darroch and Ratcliff's GIS to models without the overall effect, and, on the other hand, an extension of  IPF($\gamma$) of \cite{KRipf1} that can be applied to non-negative integer design matrices. Introducing an additional parameter $\boldsymbol \xi > \boldsymbol 0$ is a further generalization, as both original algorithms were designed for $\boldsymbol \xi \equiv \boldsymbol 1$.

Let $\mathcal{G}(\mathbf{A})$ be a general log-linear model.  Unless otherwise specified, assume that $ \mathbf{A} \equiv \mathbf{A}/\| \mathbf{A} \|_1$, where $\| \mathbf{A} \|_1$ stands for the $L_1$ norm of $\mathbf{A}$  (its maximal row sum):
$$\| \mathbf{A} \|_1 = \underset{1 \leq j \leq I}{{max}} \sum_{j = 1}^J |a_{ij}|.$$ 
Write $A_1, \dots, A_J$ for the columns of $\mathbf{A}$, and consider a $\boldsymbol q \in \mathbb{R}^I_{\geq 0}$ such that $A_1\tr \boldsymbol q, \dots, A_J\tr \boldsymbol q > 0$. Let $\mathbf{D}$ denote a kernel basis matrix of $\mathbf{A}$.

\begin{center} \textbf{Iterative Proportional Fitting Algorithm IPF($\gamma$, $\boldsymbol \xi$):} \end{center}
\begin{itemize}
\item[] \textit{Input:} a design matrix $\mathbf{A}$ and its  kernel basis matrix $\mathbf{D}$; $\boldsymbol q \in \mathbb{R}^I_{\geq 0}$; $\gamma > 0$; $\boldsymbol \xi \in \mathbb{R}^I_{\geq 0}$.
\item \textit{Initialize:} set $n = 0$; choose $\boldsymbol \delta^{(0)} \in \mathbb{R}^I_{\geq 0}$ such that $\mathbf{D} \mbox{log } \boldsymbol \delta^{(0)}  = \mathbf{D} \log \boldsymbol \xi$.
\item \textit{Iterate for } $n \geq 0$: 
\begin{eqnarray} 
\delta_i^{(n+1)} &=& \delta_i^{(n)} \prod_{j = 1}^J \left[\gamma\frac{{A}_{j}\tr \boldsymbol{q}}{{A}_{j}\tr \boldsymbol{\delta}^{(n)}}\right]^{a_{ij}}
  \,\, \mbox{for all } i \in \mathcal{I}.  \label{RipfGamma} 
\end{eqnarray}
\end{itemize}

A formal proof of convergence of the sequence (\ref{RipfGamma}) is presented as Theorem \ref{ThGammaNew} in the Appendix. It is shown that for a $\gamma > 0$, as $n \to \infty$, the sequence  $\boldsymbol{\delta}_{\gamma}^{(n)}$ converges, and its limit $\boldsymbol{\delta}_{\gamma}^{*}$ is the unique solution to the system: 
\begin{equation*}
(i)  \hspace{2mm} \mathbf{A}\tr\boldsymbol{\delta}_{\gamma}^* = \gamma \mathbf{A}\tr \boldsymbol q, \quad (ii) \hspace{2mm} \mathbf{D} \mbox{log } \boldsymbol{\delta}_{\gamma}^* = \mathbf{D} \log \boldsymbol \xi.
\end{equation*}

As corollary, it is proved that  when $\mathcal{G}(\mathbf{A})$ is either a model for probabilities with the overall effect, $\boldsymbol 1 \in colspan (\mathbf{A})$, or a  model for intensities (with or without the overall effect) the MLE under this model can be computed using IPF($\gamma$, $\boldsymbol \xi$) by setting $\gamma = 1$ and $\boldsymbol \xi = \boldsymbol 1$. 
These corollary results  are not entirely new. Firstly, if the model matrix satisfies (\ref{sumTo1}), IPF($1$, $\boldsymbol 1$) coincides with  the generalized iterative scaling of \cite{DarrochRatcliff}.
Secondly, if $\boldsymbol 1 \in colspan (\mathbf{A})$, IPF($\gamma$, $\boldsymbol \xi$) is an instantiation of MART. Thus, in these two cases, the convergence is already known.  However, because in the context where MARTs are applied the overall effect is present by design, their convergence in the no-overall-effect case had not been considered. Therefore, the results  of  Theorem \ref{ThGammaNew} concerning the convergence of IPF($\gamma$, $\boldsymbol \xi$) when $\boldsymbol 1 \notin colspan (\mathbf{A})$  are  absolutely novel.  

Notice that in the case of probabilities, $\boldsymbol \delta \equiv \boldsymbol p$, Theorem \ref{ThGammaNew} does not guarantee that the sequence limit is a probability distribution, that is, $\boldsymbol1\tr \boldsymbol{\delta}_{\gamma}^{*}  = 1$. The two-stage procedure described next incorporates an additional normalization step and is an extension of the algorithm proposed by \cite{KRipf1} for relational models.

\begin{center} \textbf{Generalized Iterative Proportional Fitting Algorithm G-IPF($\boldsymbol \xi$):} \end{center}
\begin{itemize}
\item[] \textit{Input:} a design matrix $\mathbf{A}$ and its  kernel basis matrix $\mathbf{D}$; $\boldsymbol q \in \Delta_I$;  $\boldsymbol \xi \in \mathbb{R}^I_{\geq 0}$.
\item \textit{Initialize:} set $n = 0$; $\gamma^{(0)} = 1$; choose a $\boldsymbol \delta^{(0)} > \boldsymbol 0$ for which $\mathbf{D} \log \boldsymbol \delta^{(0)} = \mathbf{D} \log \boldsymbol \xi$.
\item \textit{Core step:} for $d \geq 0$, apply IPF($\gamma^{(d)}$, $\boldsymbol \xi$) to obtain $\tilde{\boldsymbol \delta}^{(d)}_{\gamma}$ that satisfies
$\mathbf{A}\tr \tilde{\boldsymbol \delta}_{\gamma^{(d)}} =  \gamma^{(d)} \mathbf{A}\tr\boldsymbol q$ and 
$\mathbf{D}\log \tilde{\boldsymbol \delta}_{\gamma^{(d)}} = \mathbf{D}\log {\boldsymbol \xi}$. 
\item \textit{Adjustment step:}  adjust $\gamma^{(d)}$ in such a way that after the $(d+1)$th Core step, the total of the limit $\tilde{\boldsymbol \delta}_{\gamma^{(d +1)}}$ is closer to 1 than the total of $\tilde{\boldsymbol \delta}_{\gamma^{(d)}}$:
$$|\boldsymbol 1\tr \tilde{\boldsymbol \delta}_{\gamma^{(d +1)}} -1 | < |\boldsymbol 1\tr \tilde{\boldsymbol \delta}_{\gamma^{(d)}} -1 |.$$
\item \textit{Iterate between Core step and Adjustment step} until $|\boldsymbol 1\tr \tilde{\boldsymbol \delta}_{\gamma^{(d +1)}} -1 |$ is smaller than a desired precision.
\end{itemize}

Theorem \ref{ThGammaNew}  and a result of \cite{Forcina2019} imply that
as $d \to \infty$, the sequence $\tilde{\boldsymbol \delta}_{\gamma^{(d)}}$ obtained from G-IPF ($\boldsymbol \xi$) converges, its limit, $\boldsymbol \delta^{(*)}$, is the unique solution to the system:
$$\mathbf{A}\tr \boldsymbol \delta^{(*)} =  \gamma^{(*)} \mathbf{A}\tr\boldsymbol q, \,\,  
\mathbf{D}\log\boldsymbol \delta^{(*)} = \mathbf{D}\log {\boldsymbol \xi}, \,\, \boldsymbol 1\tr \boldsymbol \delta^{(*)} = 1,$$
for a unique $\gamma^{(*)} >0$.  Therefore, the G-IPF($\boldsymbol \xi$) algorithm can be used to generate distributions  from  $\mathcal{G}_{1}(\boldsymbol \xi)$.  

Finally, as seen in the proof of Theorem \ref{ThGammaNew},  the sequence generated by G-IPF is built by successive projections with respect to the Bregman information divergence \citep{Bregman}. This result generalizes the geometric interpretation associated to IPF and GIS  given by  \cite{Csiszar, CsiszarGIS}.

In the next section, the geometric framework for power analysis is applied to  real data.

\section{An example of geometric power analysis}
\label{SectionExample}

The study Dia-Vacc (NCT04799808) was initiated in January 2021 in several nephrology centers in  Germany \citep{DIAVacc}.  The participants were initially vaccinated with two doses either of BNT162b2mRNA or 1273-mRNA vaccine, and those who developed the vaccine-specific antibodies were said to have a positive immune response to the vaccination. The study participants without a positive response were revaccinated, \textit{boostered}, after six months, and those who did not respond to the booster,  were given another one later. As a result, the study participants can be classified into four groups according to their vaccination profile, see Example \ref{VaccineModel} and   Figure \ref{TreeVaccine}. Denote by $\boldsymbol p = $ ($p_1$, $p_2$, $p_3$, $p_4)\tr$ the probabilities of corresponding paths, by  $\theta_1$ the probability of positive response and $\theta_0= 1-\theta_1$ the probability of non-response to a single vaccination. Then, the hypothesis of independence between outcomes of subsequent vaccinations, or in other words, of the non-existence of delayed immune response, can be expressed using the model (\ref{calves3}):
\begin{equation}\tag{\ref{calves3}}
p_{1} = \theta_0^3, \,\,  p_{2} = \theta_0^2 \theta_1, \,\, p_{3} = \theta_0\theta_1, \,\, p_{4} = \theta_1. 
\end{equation}
In this manuscript, only the data of kidney transplant recipients participating in the study will be analyzed. After three vaccination rounds, the observed frequencies were $\boldsymbol y$ = $(80,$ $12,$ $44,$ $64)\tr$ in each respective category.  As obtained in (\ref{MLEcalves3}), the MLEs for probabilities under the model (\ref{calves3}) are equal to 
$\hat{\boldsymbol p}$ = $((308/428)^3$, $308^2 \cdot 120/428^3$, $308 \cdot 120/428^2$, $120/428)\tr$ $\approx (0.373,$
$0.145,$ $0.202,$ $0.280)\tr$. The Pearson statistic $X^2 \approx 11.85$  and deviance $G^2 \approx 14.65$, on two degrees of freedom, indicate medium-strength evidence against the hypothesis of independence.  Assuming the significance level of $5\%$, the achieved (conventional) statistical power is about $88\%$ \citep{SilveyLMtest, AitchSilvey60}. 

Using the Monte-Carlo method of Section \ref{SectionPower}, two posteriori MC-based cumulative power estimates were obtained. The distributions on the simplex were drawn from the distribution $Dir_4(1)$ and $Dir_4(1/2)$, respectively. In each case, ten Monte-Carlo sequences of $10^4$ replications were generated. The average rejection rates were equal to $0.903$ and $0.845$, with the $95\%$ confidence intervals $(0.901, 0.905)$ and $(0.841, 0.849)$, respectively. The convergence of three of those sequences (for the sake of graphical clarity) is illustrated in Figures \ref{MCMCvaccine1} and \ref{MCMCvaccine2}.

Suppose a researcher would like to conduct a study with a similar design. The question about the study size needed to reach a prespecified power in rejecting the null hypothesis (\ref{calves3}) using the chi-square test will be addressed below using the geometric power framework. In order to carry out the apriori Monte-Carlo power calculations, one starts with specifying the null and alternative hypotheses in a dual form, using a set of constraints on odds ratios. For the hypothesis of independence (\ref{calves3}), a dual form can be obtained from the  matrix 
$$\mathbf{D} = \left( \begin{array}{rrrr}
                                   1 & -2 & 1 & 1\\
                                   0 & 1 & -2 & 1\\
                                   \end{array} \right).
                                 $$ 
In this case, the null hypothesis is formulated as: 
\begin{equation}\label{H0vaccine}
\mathcal{H}_0: \boldsymbol \pi \in \mathcal{G}_{0} = \{\boldsymbol p \in  \Delta_4: \,  p_1 p_3 p_4 / p_2^2 = 1, \quad  p_2 p_4/ p_3^2 =  1\}.
\end{equation}
As an alternative, one can take:
\begin{equation}\label{H1vaccine}
\mathcal{H}_1: \boldsymbol \pi \in \mathcal{G}_1 = \{\boldsymbol p \in  \Delta_4: \,  p_1 p_3 p_4 / p_2^2 = 1, \quad  p_2 p_4 / p_3^2 = k\},
\end{equation}
for a fixed $k > 0$. Power calculations for $k = 2$ and $k = 3$ are presented in Table \ref{PowerTableVaccine1}. According to the simulation results, in order to reach a power of $80\%$ (under a $5\%$ significance level), a sample size of about 490 would be needed for the alternative with $k = 2$ and of about 210 for the alternative with $k = 3$.  

For comparison, notice that in the observed data above, the sample size is $N = 200$ and $\boldsymbol p = (0.4, 0.06, 0.22, 0.32)$, and therefore, $p_1 p_3 p_4 / p_2^2 \approx 7.822$ and $p_2 p_4 / p_3^2 \approx 0.397$. Considering the magnitude of the observed ratios, the classical posteriori power estimate of $88 \%$ and the cumulative power estimates of $90\%$ and $84\%$ seem to be in a reasonable agreement with the power-sample size summary in Table \ref{PowerTableVaccine1}. \qed

\section{Conclusion}\label{SectionConclusion}

The statistical power analysis for log-linear models is routinely conducted using the noncentral chi-square distribution, with the goodness-of-fit index performing the role of effect size. This method overlooks the phenomenon that the same goodness-of-fit index corresponds to a variety of association structures expressed in terms of odds ratios.  In the context of power analysis, this property may lead to contradicting interpretations of the proximity between the null hypothesis and its alternative. This paper focuses on the power analysis of goodness-of-fit of a log-linear model  against a log-affine alternative. In this case, the deviation of the alternative hypothesis from the null one can be quantified exclusively using odds ratios. The power analysis is given a geometric interpretation, and a concept of geometric power is proposed.  A Monte-Carlo method for power analysis is described and demonstrated using several examples. An additional original contribution of the paper is the new proof of convergence of generalized interactive scaling for log-linear models. This result extends a previously known geometric interpretation of this algorithm as well. All computations were performed using the R Environment for statistical computing \citep{Rcite}.

\section*{Acknowledgements}

The author is grateful to Tam\'{a}s Rudas for his advice during the preparation of this manuscript. Tables \ref{ExTR1} and \ref{ExTR2} were shown during the joint presentation of Tam\'{a}s Rudas and the first author at the ERCIM-11 conference in London. The author would like to especially thank the associate editor and two anonymous reviewers for their suggestions about the initially submitted version of the manuscript.
The data set is courtesy of Christian Hugo and Julian Stumpf, from the University Clinic Carl-Gustav Carus, Technical University Dresden.

\section*{Conflict of interest}

On behalf of all authors, the corresponding author states that there is no conflict of interest.

\bibliography{KlimovaGeometryBIB20240422.bib}
\bibliographystyle{apacite}


\clearpage
 
\section*{Appendix: Convergence of  the IPF($\gamma$, $\boldsymbol \xi$) algorithm}

\setcounter{section}{4}



\begin{center} \textbf{Algorithm} IPF($\gamma$, $\boldsymbol \xi$): \end{center}
\begin{itemize}
\item[] \textit{Input:} $\gamma > 0$; $\boldsymbol \xi \in \mathbb{R}^I_{\geq 0}$; $\boldsymbol \delta^{(0)} \in \mathbb{R}^I_{\geq 0}$ such that $\mathbf{D} \mbox{log } \boldsymbol \delta^{(0)}  = \mathbf{D} \log \boldsymbol \xi$.
\item[] \textit{Iterate for } $n \geq 0$: 
\begin{equation} 
\delta_i^{(n+1)} = \delta_i^{(n)} \prod_{j = 1}^J \left[\gamma\frac{{A}_{j} \boldsymbol{q}}{{A}_{j} \boldsymbol{\delta}^{(n)}}\right]^{a_{ij}}
  \,\, \mbox{for all } i \in \mathcal{I}. \tag{\ref{RipfGamma}} 
\end{equation}
\end{itemize}

\vspace{3mm}

\noindent For $\boldsymbol t, \boldsymbol u \in \mathbb{R}^{|\mathcal{I}|}_{>0}\,$, let $\mathcal{D}(\boldsymbol t, \boldsymbol u)$ denote the Bregman divergence between $\boldsymbol t$ and $\boldsymbol u$:
\begin{equation} \label{BDdef} 
\mathcal{D}(\boldsymbol t, \boldsymbol u) = \sum_{i \in \mathcal{I}} t_i \mbox{log }(t_i/u_i) - (\sum_{i \in \mathcal{I}} t_i - \sum_{i \in \mathcal{I}} u_i).
\end{equation}
Then, for the sequence (\ref{RipfGamma}) obtained in IPF($\gamma$, $\boldsymbol \xi$) the following holds:

 \begin{lemma}\label{LemmaBregmanInequality}
 Set $\gamma = 1$. Then,
  $$\sum_{i = 1}^I \delta^{(n+1)}_i \leq \sum_{i = 1}^I \delta^{(n)}_i +\sum_{j = 1}^J (A_j\tr \boldsymbol q - A_j\tr\boldsymbol \delta^{(n)}),$$
  and for any $\boldsymbol z \in \mathbb{R}^I_{>0}$,
  $$\mathcal{D}(\boldsymbol z, \boldsymbol \delta^{(n+1)}) \leq  \mathcal{D}(\boldsymbol z, \boldsymbol \delta^{(n)}) - \mathcal{D}(\mathbf{A}\tr\boldsymbol q, \mathbf{A}\tr\boldsymbol \delta^{(n)}).$$
 \end{lemma}

\begin{proof}

The first statement:  
 \vspace{3mm}
  
First recall that, for $x_1, \dots x_n > 0$, and $w_i \geq 0$, such that $w_1 + \cdots + w_n = w$, according to
the weighted AM-GM (arithmetic mean - geometric mean) inequality, one has:
\begin{equation}\label{wAGinequality}
 \frac{w_1 x_1 + \dots w_n x_n}{w} \geq \left (x_1^{w_1} \cdots x_n^{w_n}  \right)^{1/w}.
 \end{equation}

Notice that $$\delta_i^{(n+1)} = \delta_i^{(n)} \prod_{j = 1}^J \left[\frac{{A}_{j}\tr \boldsymbol{q}}{{A}_{j}\tr \boldsymbol{\delta}^{(n)}}\right]^{a_{ij}} \cdot 1^{1-\sum_{j=1}^J a_{ij}},$$
and apply the inequality (\ref{wAGinequality}):
$$\delta_i^{(n+1)} \leq \delta_i^{(n)} \left( \sum_{j = 1}^J a_{ij} \frac{{A}_{j}\tr \boldsymbol{q}}{{A}_{j}\tr \boldsymbol{\delta}^{(n)}}  + {1-\sum_{j=1}^J a_{ij}}\right).$$   
After regrouping terms:                                  
$$\delta_i^{(n+1)} \leq \sum_{j = 1}^J a_{ij}\delta_i^{(n)}  \frac{{A}_{j}\tr \boldsymbol{q}}{{A}_{j}\tr \boldsymbol{\delta}^{(n)}}  + \delta_i^{(n)} -\sum_{j=1}^J a_{ij}\delta_i^{(n)}.$$                                             
Therefore,
\begin{align*}
\sum_{i = 1}^I\delta_i^{(n+1)} &\leq \sum_{i = 1}^I \sum_{j = 1}^J a_{ij}\delta_i^{(n)}  \frac{{A}_{j}\tr \boldsymbol{q}}{{A}_{j}\tr \boldsymbol{\delta}^{(n)}}  +  \sum_{i =1}^I\delta_i^{(n)} -  \sum_{i =1}^I\sum_{j=1}^J a_{ij}\delta_i^{(n)}\\
& = \sum_{i =1}^I\delta_i^{(n)} +  \sum_{j = 1}^J   \frac{{A}_{j}\tr \boldsymbol{q}}{{A}_{j}\tr \boldsymbol{\delta}^{(n)}} \sum_{i = 1}^I a_{ij}\delta_i^{(n)} - \sum_{j = 1}^J A_j\tr \boldsymbol \delta^{(n)} \\
& = \sum_{i =1}^I\delta_i^{(n)} +  \sum_{j = 1}^J   \frac{{A}_{j}\tr \boldsymbol{q}}{{A}_{j}\tr \boldsymbol{\delta}^{(n)}} A_j\tr \boldsymbol \delta^{(n)} - \sum_{j = 1}^J A_j\tr \boldsymbol \delta^{(n)}  = \sum_{i =1}^I\delta_i^{(n)} + \sum_{j = 1}^J (A_j\tr \boldsymbol q - A_j\tr\boldsymbol \delta^{(n)}).
\end{align*}                                           
 
The second statement:
\vspace{3mm}

For a $\boldsymbol z \in \mathbb{R}^I_{>0}$, one has:
 \begin{align*}
 \mathcal{D}(\boldsymbol z, \boldsymbol \delta^{(n+1)}) &= \sum_{i} z_i \log z_i -  \sum_{i} z_i \log \delta_i^{(n +1)} - (\sum_{i} z_i  - \sum_{i} \delta_i^{(n +1)}) \\
 &\leq \sum_{i} z_i \log z_i -  \sum_{i} z_i \log \delta_i^{(n +1)} - \sum_{i} z_i  + \sum_{i =1}^I\delta_i^{(n)} +  \sum_{j = 1}^J   {A}_{j}\tr \boldsymbol{q} - \sum_{j = 1}^J A_j\tr \boldsymbol \delta^{(n)}  \\
 &= \sum_{i} z_i \log z_i  - \sum_{i} z_i \log  \delta_i^{(n)} \prod_{j = 1}^J \left[\frac{{A}_{j}\tr \boldsymbol{q}}{{A}_{j}\tr \boldsymbol{\delta}^{(n)}}\right]^{a_{ij}} \\
 &- \sum_{i} z_i  + \sum_{i =1}^I\delta_i^{(n)} +  \sum_{j = 1}^J   {A}_{j}\tr \boldsymbol{q} - \sum_{j = 1}^J A_j\tr \boldsymbol \delta^{(n)}\\
 &= \sum_{i} z_i \log z_i  - \sum_{i} z_i \log  \delta_i^{(n)} -   \sum_{i} z_i \sum_{j = 1}^J{a_{ij}}  \log \left[\frac{{A}_{j}\tr \boldsymbol{q}}{{A}_{j}\tr \boldsymbol{\delta}^{(n)}}\right] \\
 &- \sum_{i} z_i  + \sum_{i =1}^I\delta_i^{(n)} +  \sum_{j = 1}^J   {A}_{j}\tr \boldsymbol{q} - \sum_{j = 1}^J A_j\tr \boldsymbol \delta^{(n)}\\
 &= \sum_{i} z_i \log (z_i/\delta_i^{(n)}) - (\sum_{i} z_i  - \sum_{i =1}^I\delta_i^{(n)})-   \sum_{i} z_i \sum_{j = 1}^J{a_{ij}}  \log \left[\frac{{A}_{j}\tr \boldsymbol{q}}{{A}_{j}\tr \boldsymbol{\delta}^{(n)}}\right] \\
 &+  \sum_{j = 1}^J   {A}_{j}\tr \boldsymbol{q} - \sum_{j = 1}^J A_j\tr \boldsymbol \delta^{(n)}\\
 &= \mathcal{D}(\boldsymbol z, \boldsymbol \delta^{(n)}) - 
 \sum_{j = 1}^J   \log \left[\frac{{A}_{j}\tr \boldsymbol{q}}{{A}_{j}\tr \boldsymbol{\delta}^{(n)}}\right]\sum_{i} z_i{a_{ij}} + \sum_{j = 1}^J   {A}_{j}\tr \boldsymbol{q} - \sum_{j = 1}^J A_j\tr \boldsymbol \delta^{(n)}  \\
 & = \mathcal{D}(\boldsymbol z, \boldsymbol \delta^{(n)}) - 
 \sum_{j = 1}^J  A_j\tr \boldsymbol z \cdot \log \left[\frac{{A}_{j}\tr \boldsymbol{q}}{{A}_{j}\tr \boldsymbol{\delta}^{(n)}}\right] + \sum_{j = 1}^J   {A}_{j}\tr \boldsymbol{q} - \sum_{j = 1}^J A_j\tr \boldsymbol \delta^{(n)} .
  \end{align*}
 
 Therefore, after taking $\boldsymbol z \equiv \boldsymbol q$,
 $$\mathcal{D}(\boldsymbol q, \boldsymbol \delta^{(n+1)}) \leq  \mathcal{D}(\boldsymbol q, \boldsymbol \delta^{(n)}) - \mathcal{D}(\mathbf{A}\tr\boldsymbol q, \mathbf{A}\tr\boldsymbol \delta^{(n)}).$$
 
\end{proof}

\vspace{3mm}

\begin{theorem}\label{ThGammaNew}
Fix a $\gamma > 0$. Then, as $n \to \infty$, the sequence  $\boldsymbol{\delta}_{\gamma}^{(n)}$ converges, and its limit $\boldsymbol{\delta}_{\gamma}^{*}$ is the unique solution to the system: 
\begin{tabbing}
(i) \= \hspace{2mm} $\mathbf{A}\tr\boldsymbol{\delta}_{\gamma}^* = \gamma \mathbf{A}\tr \boldsymbol q$, \\
(ii)\> \hspace{3mm} $\mathbf{D} \mbox{log } \boldsymbol{\delta}_{\gamma}^* = \mathbf{D} \log \boldsymbol \xi$.\\
\end{tabbing}
\end{theorem}    
 
\begin{proof}

(i) First, let $\gamma =1$. By Lemma \ref{LemmaBregmanInequality}, the sequence $\mathcal{D}(\boldsymbol q, \boldsymbol \delta^{(n)})$  is monotone decreasing. Therefore, because $\mathcal{D}(\boldsymbol q, \boldsymbol \delta^{(n)}) \geq 0$,  one has that $\mathcal{D}(\mathbf{A}\tr\boldsymbol q, \mathbf{A}\tr\boldsymbol \delta^{(n)}) \to 0$ as $n \to \infty$. By the properties of Bregman divergence, the latter implies that 
 $\mathbf{A}\tr\boldsymbol \delta^{(n)} \to \mathbf{A}\tr\boldsymbol q$ \citep{Bregman}.

 In more generality, when $0 < \gamma \neq 1$,  after replacing $\boldsymbol q$ with $\gamma \boldsymbol q$, and, respectively, $\boldsymbol{\delta}^{(n)}$ with $\boldsymbol{\delta}_{\gamma}^{(n)}$, one has 
   $$\mathcal{D}(\boldsymbol z, \boldsymbol \delta^{(n+1)}) \leq \mathcal{D}(\boldsymbol z, \boldsymbol \delta^{(n)}) - 
 \sum_{j = 1}^J  A_j\tr \boldsymbol z \cdot \log \left[\frac{\gamma{A}_{j}\tr \boldsymbol{q}}{{A}_{j}\tr \boldsymbol{\delta}^{(n)}}\right] + \sum_{j = 1}^J  \gamma {A}_{j}\tr \boldsymbol{q} - \sum_{j = 1}^J A_j\tr \boldsymbol \delta^{(n)},$$
 and for a $\boldsymbol z \in \mathbb{R}^I_{>0}$, that $\mathbf{A}\tr \boldsymbol z = \gamma \mathbf{A}\tr \boldsymbol q$:
 $$\mathcal{D}(\boldsymbol z, \boldsymbol \delta^{(n+1)}) \leq  \mathcal{D}(\boldsymbol q, \boldsymbol \delta^{(n)}) - \mathcal{D}(\gamma \mathbf{A}\tr\boldsymbol q, \mathbf{A}\tr\boldsymbol \delta^{(n)}).$$
Analogous to the above, $\mathbf{A}\tr\boldsymbol \delta^{(n)} \to \gamma \mathbf{A}\tr\boldsymbol q$, so one concludes that the sequence $\boldsymbol{\delta}_{\gamma}^{(n)}$ converges and its limit, $\boldsymbol{\delta}_{\gamma}^*$, satisfies 
$$\mathbf{A}\tr\boldsymbol{\delta}_{\gamma}^* = \gamma\boldsymbol A\tr \boldsymbol{q}.$$ 
(ii) The proof is by induction. Since $\mathbf{D} \mbox{log } \boldsymbol {\delta}_{\gamma}^{(0)} = \mathbf{D} \log \boldsymbol \xi$,  the statement holds for $n = 0$. 

Further, suppose $\mathbf{D} \mbox{log } \boldsymbol{\delta}_{\gamma}^{(d)} = \mathbf{D} \log \boldsymbol \xi$ for an $n > 0$. Set $C_j = \frac{\gamma {A}_{j}\tr \boldsymbol{q}}{{A}_{j}\tr \boldsymbol{\delta}_{\gamma}^{(n)}}$. Then, 
\begin{eqnarray*}
\mathbf{D} \mbox{log } \boldsymbol{\delta}_{\gamma}^{(n+1)} &=& \mathbf{D} \mbox{log }\left[\begin{array}{c}  
{\delta}_{\gamma,1}^{(n)}\cdot \prod_{j = 1}^J C_j^{a_{1j}}\\
{\delta}_{\gamma,2}^{(n)}\cdot\prod_{j = 1}^J C_j^{a_{2j}}\\
\vdots \\
{\delta}_{\gamma,I}^{(n)}\cdot \prod_{j = 1}^J C_j^{a_{Ij}}
\end{array}\right] = \mathbf{D} \left[\begin{array}{c}  
\mbox{log }{\delta}_{\gamma,1}^{(n)}+ \sum_{j = 1}^J{a_{1j}}\mbox{log } C_j\\
\mbox{log }{\delta}_{\gamma,2}^{(n)}+ \sum_{j = 1}^J{a_{2j}}\mbox{log } C_j\\
\vdots \\
\mbox{log }{\delta}_{\gamma,I}^{(n)}+\sum_{j=1}^J{a_{Ij}}\mbox{log } C_j
\end{array}\right] \\
& & \\
&=&  \mathbf{D} \mbox{log } \boldsymbol{\delta}_{\gamma}^{(n)} + \sum_{j = 1}^J\mbox{log } C_j \mathbf{D}A_j = \mathbf{D} \log \boldsymbol \xi,
\end{eqnarray*}
because $\mathbf{D}$ is a kernel basis matrix and thus $\mathbf{D}A_{j}= \boldsymbol 0$.
Therefore, $\mathbf{D} \mbox{log } \boldsymbol{\delta}_{\gamma}^{(n)} = \mathbf{D} \log \boldsymbol \xi$ for all $n \geq 0$, and, by continuity of matrix multiplication and logarithm, $\mathbf{D} \mbox{log } \boldsymbol{\delta}_{\gamma}^{*} = \mathbf{D} \log \boldsymbol \xi$. This completes the proof.
 \end{proof}


\begin{corollary} \label{G_IPFepsConvCorP}
Let $\mathcal{G}(\mathbf{A})$ be a general log-linear model for intensities whose design matrix $\mathbf{A}$ satisfies the properties specified in the beginning of this section. Suppose $\boldsymbol y$ is a realization of $Pois(\boldsymbol \lambda)$ and assume that the MLE $\hat{\boldsymbol \lambda}$ given $\boldsymbol y$ exists. 
The sequence
$\{\boldsymbol{\lambda}^{(n)}\}  \in \mathbb{R}^I_{\geq 0}$ obtained as:
\begin{align} 
&\boldsymbol \lambda^{(n)} = \boldsymbol 1, \nonumber \\
&\lambda_i^{(n+1)} = \lambda_i^{(n)} \prod_{j = 1}^J \left[\frac{{A}_{j}\tr \boldsymbol{y}}{{A}_{j}\tr \boldsymbol{\lambda}^{(n)}}\right]^{a_{ij}}
  \,\, \mbox{for all } i \in \mathcal{I}, n \geq 0,   \label{RipfGammaP} 
\end{align}
converges and its limit is $\boldsymbol \lambda^{(*)}$ is equal to the MLE   $\hat{\boldsymbol \lambda}$ under the model $\mathcal{G}(\mathbf{A})$.
 \end{corollary}

\begin{proof}

\noindent Apply Theorem \ref{ThGammaNew} to $\boldsymbol q \equiv \boldsymbol y$, $\boldsymbol \xi \equiv \boldsymbol 1$, $\gamma = 1$. The limit satisfies (\ref{MLEsysP}), as required.
\end{proof}

\vspace{3mm}

\begin{corollary} \label{G_IPFepsConvCorM}
Let $\mathcal{G}(\mathbf{A})$ be a general log-linear model for probabilities with $\boldsymbol 1 \in colspan (\mathbf{A})$. Suppose $\boldsymbol y$ is a realization of $Mult(N, \boldsymbol p)$ and assume that the MLE $\hat{\boldsymbol p}$ given $\boldsymbol y$ exists. Let $\boldsymbol q = \boldsymbol y/N$. Then, the sequence
$\{\boldsymbol{p}^{(n)}\}  \in \mathbb{R}^I_{\geq 0}$ obtained as:
\begin{align} 
&\boldsymbol p^{(n)} = \boldsymbol 1, \nonumber \\
&p_i^{(n+1)} = p_i^{(n)} \prod_{j = 1}^J \left[\frac{{A}_{j}\tr \boldsymbol{q}}{{A}_{j}\tr \boldsymbol{p}^{(n)}}\right]^{a_{ij}}
  \,\, \mbox{for all } i \in \mathcal{I}, n \geq 0,   \label{RipfGammaM} 
\end{align}
converges and its limit is $\boldsymbol p^{(*)}$ is equal to the MLE   $\hat{\boldsymbol p}$ under the model $\mathcal{G}(\mathbf{A})$.
 \end{corollary}

\begin{proof}
\noindent Taking $\boldsymbol \xi \equiv \boldsymbol 1$, $\gamma = 1$  in Theorem \ref{ThGammaNew}, one sees that the sequence limit satisfies the first two equations in (\ref{MLEsysM}). Because $\boldsymbol 1\in colspan(\mathbf{A})$ and $\boldsymbol 1\tr \boldsymbol q = 1$, it also holds that   $\boldsymbol 1\tr \boldsymbol p^{(*)} = 1$. 
\end{proof}

\clearpage
 
\section*{Supplementary Material}

\begin{figure}[ht]
\begin{center}
\includegraphics[scale=0.5]{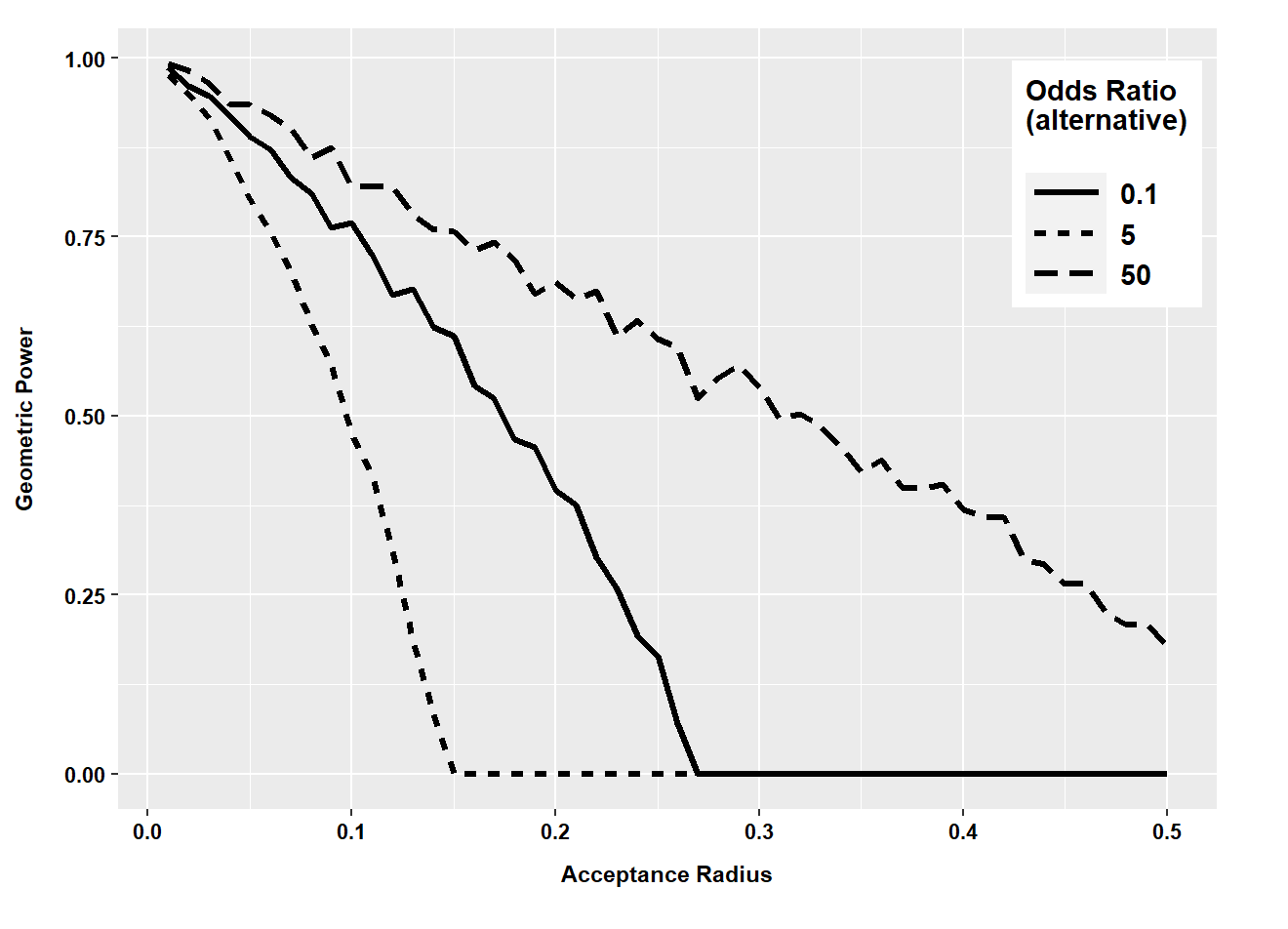}
\end{center}
\caption{Geometric power of the goodness-of-fit test of $2\times 2$ independence for different acceptance radii. Samples from $\Delta_4$ were drawn using $Dir_4(1)$}
\label{PowerIndepUniform}
\end{figure}

\begin{figure}
\begin{center}
\includegraphics[scale=0.5]{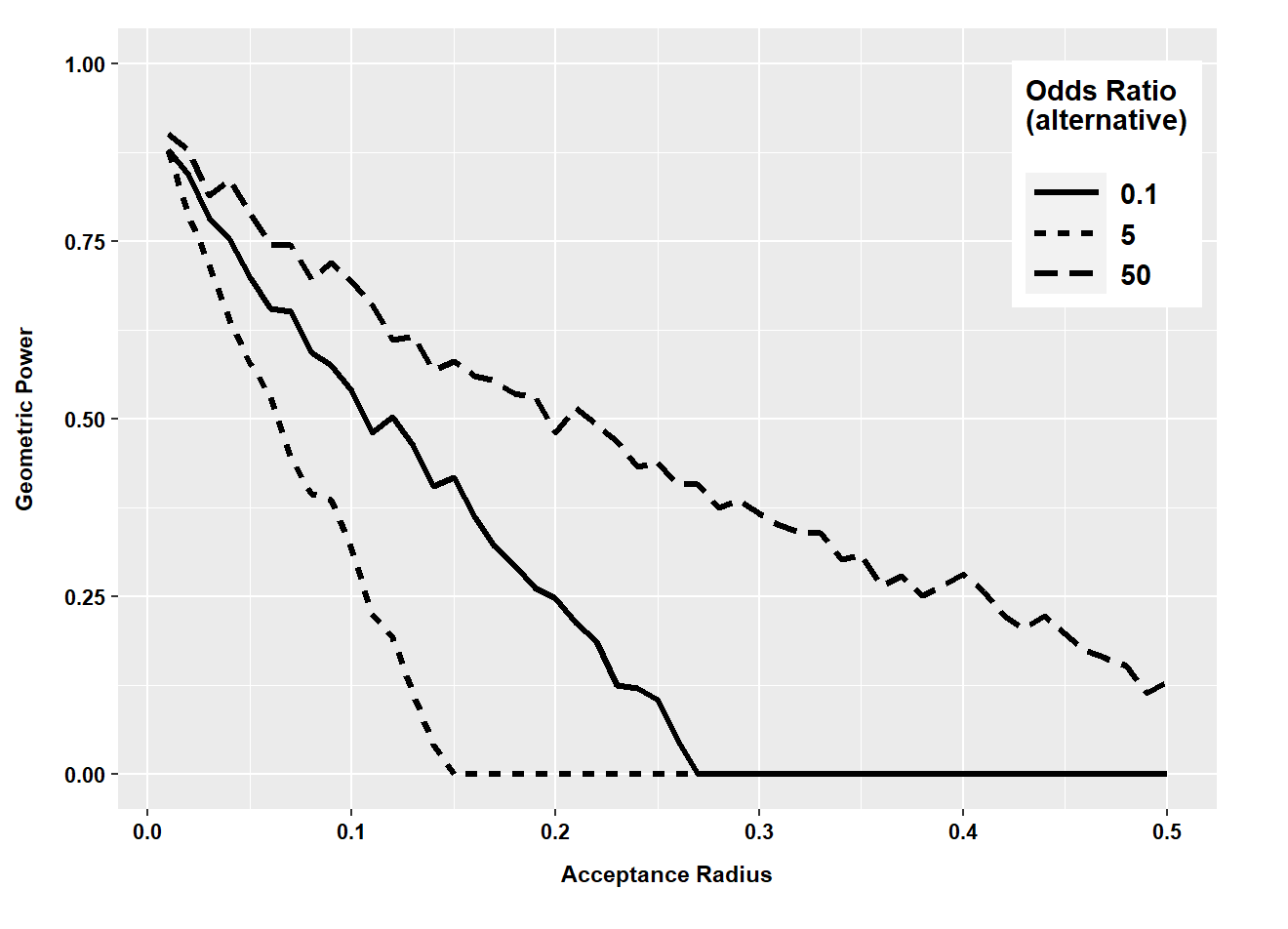}
\end{center}
\caption{Geometric power of the goodness-of-fit test of $2\times$ independence for different acceptance radii. Samples from $\Delta_4$ were drawn $Dir_4(1/2)$}
\label{PowerIndepJeffreys}
\end{figure}

\begin{figure}
\begin{center}
\includegraphics[scale=0.5]{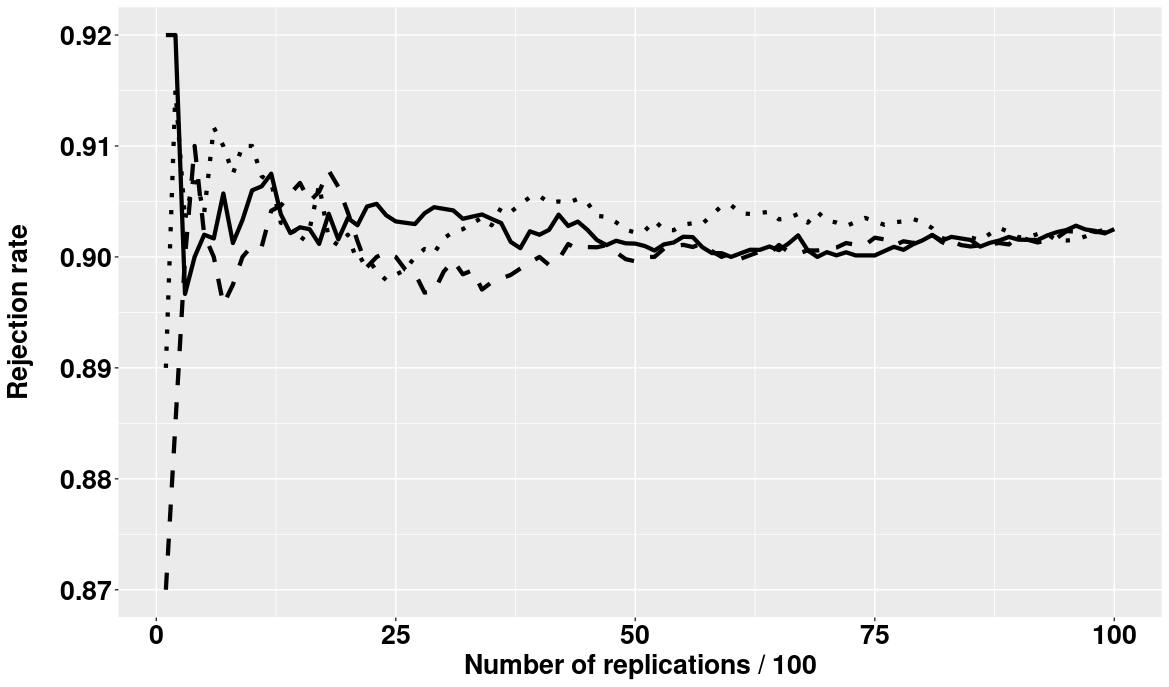}
\end{center}
\caption{Convergence of rejection rates of MCMC generated using $Dir_4(1)$}
\label{MCMCvaccine1}
\end{figure}

\begin{figure}
\begin{center}
\includegraphics[scale=0.5]{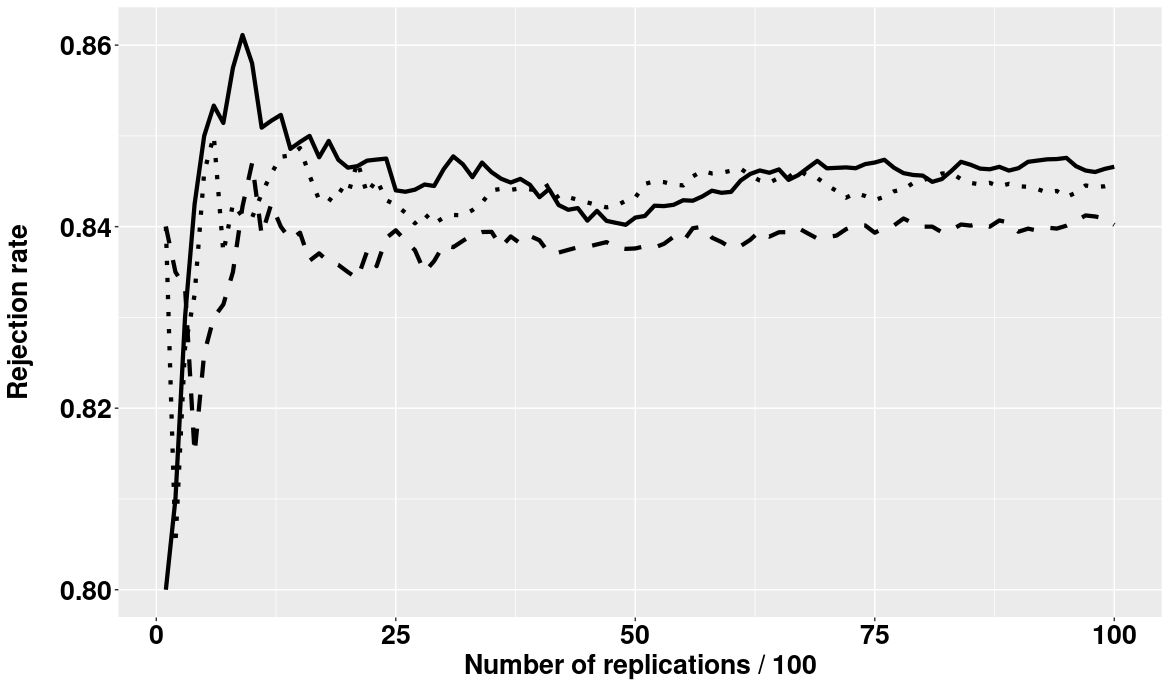}
\end{center}
\caption{Convergence of rejection rates of MCMC generated using $Dir_4(1/2)$}
\label{MCMCvaccine2}
\end{figure}

\begin{table}[ht]
\centering
\small{
\begin{tabular}{r|rr|rr||rr|rr}
&\multicolumn{4}{c||}{Sampling from $Dirichlet_4(1)$} &\multicolumn{4}{c}{Sampling from $Dirichlet_4(1/2)$} \\ [6pt]
\cline{2-9}
&\multicolumn{2}{c|}{$\mathcal{H}_1 =\left\{p_1 p_3 p_4 / p_2^2 = 1,\right.$} &\multicolumn{2}{c||}{$\mathcal{H}_1 = \left\{p_1 p_3 p_4 / p_2^2 = 1,\right.$} & 
\multicolumn{2}{c|}{$\mathcal{H}_1 =\left\{p_1 p_3 p_4 / p_2^2 = 1,\right.$} &\multicolumn{2}{c}{$\mathcal{H}_1 = \left\{p_1 p_3 p_4 / p_2^2 = 1,\right.$}
\\ [6pt]
&\multicolumn{2}{c|}{$\left.p_2 p_4 / p_3^2 = 2\right\}$} &\multicolumn{2}{c||}{$\left.p_2 p_4 / p_3^2 = 3 \right\}$}  &\multicolumn{2}{c|}{$\left.p_2 p_4 / p_3^2 = 2\right\}$} &\multicolumn{2}{c}{$\left.p_2 p_4 / p_3^2 = 3 \right\}$}\\[6pt]
\cline{2-9}
$N$ & $\alpha = 0.05$ & $\alpha = 0.10$ & $\alpha = 0.05$ & $\alpha = 0.10$& $\alpha = 0.05$ & $\alpha = 0.10$ & $\alpha = 0.05$ & $\alpha = 0.10$\\
\hline

   200 & 0.45 & 0.59 & 0.84 & 0.90 & 0.43 & 0.55 & 0.80 &  0.87\\

   220 & 0.49 & 0.62 & 0.87 & 0.92 & 0.47 & 0.60& 0.83 & 0.89 \\  

   240 & 0.53 & 0.66 & 0.90 & 0.94 & 0.50& 0.63  & 0.85 & 0.90 \\  

   260 & 0.57 & 0.69 & 0.91 &0.95 & 0.55& 0.67  & 0.88 & 0.92\\  

   280 & 0.60 & 0.72 & 0.94 & 0.96 & 0.58& 0.70  & 0.90 & 0.93\\

   300 & 0.64 & 0.75 & 0.94 &0.97 & 0.61& 0.72  & 0.91 & 0.94\\ 

   320 & 0.67 & 0.78 & 0.95 & 0.97& 0.63& 0.74  & 0.92 & 0.95\\ 

   340 & 0.70 & 0.79  & 0.97 & 0.98 & 0.66& 0.76  & 0.92 & 0.95\\ 

   360 & 0.73 & 0.83  & 0.97 & 0.98 & 0.70& 0.79  & 0.93 & 0.96\\

   380 & 0.75& 0.84 & 0.98 & 0.99 & 0.71&0.80  & 0.94 & 0.96\\

   400 & 0.77 & 0.85 & 0.98 & 0.99 & 0.73&0.82  & 0.94 & 0.96\\

   420 & 0.79 & 0.86 & 0.98 & 0.99 & 0.75&0.83  & 0.96 & 0.97\\ 

   440 & 0.81 & 0.88 & 0.98 & 0.99 & 0.77&0.85  & 0.96 & 0.97\\ 

   460 & 0.83 & 0.89 & 0.99 & 0.99& 0.79&0.86  & 0.96 & 0.97\\

   480 & 0.85 & 0.90 & 0.99 & 0.99 & 0.80&0.87  & 0.96 & 0.97\\ 
   500 & 0.86 & 0.91& 0.99 & 0.99 & 0.82&0.88  & 0.96 & 0.98\\
\end{tabular}
}
\caption{Rejection rates of $\mathcal{H}_0$ specified in (\ref{H0vaccine}) against the alternative $\mathcal{H}_1$ in (\ref{H1vaccine}) for $k = 2$ and $k = 3$, respectively. Sampling from $\Delta_4$ implemented using Dirichlet distribution $Dir_4(1)$ or  $Dir_4(1/2)$}
\label{PowerTableVaccine1}
\end{table}

\end{document}